\documentclass[11pt,a4paper]{article}

\usepackage[hidelinks, pagebackref=true]{hyperref}

\usepackage{authblk}
\usepackage{fullpage}
\usepackage{amsthm}
\usepackage{amsfonts}
\usepackage{amsmath}
\usepackage{subcaption}
\usepackage{amssymb}
\usepackage{comment}

\usepackage{cleveref}

\title{Adapting Stable Matchings to Forced and Forbidden Pairs} 

\author{Niclas Boehmer}
\author{Klaus Heeger}

\affil{\small
  Technische Universit\"at Berlin, Faculty~IV, Algorithmics and Computational 
  Complexity\protect\\
  \{niclas.boehmer,heeger\}@tu-berlin.de}

\date{\today}

\usepackage{etoolbox}

\newcommand{\lv}[1]{}

\newcommand{\appmark}{$\star$}

\usepackage{subcaption}
\usepackage{cleveref}
\usepackage[inline]{enumitem}
\usepackage{xspace}

\usepackage[inline]{enumitem}

\usepackage[textwidth=20mm,obeyFinal,colorinlistoftodos]{todonotes}
\usepackage{marginnote}

\usepackage{multirow}
\usepackage{booktabs}
\usepackage[inline]{enumitem}
\usepackage{tabularx,environ}

\newcommand{\appsymb}{$\bigstar$}
\newcommand{\appref}[1]{{\hyperref[proof:#1]{\appsymb}}}
\usepackage{thmtools}
\usepackage{thm-restate}

\newtheorem{proposition}{Proposition}
\newtheorem{observation}{Observation}
\newtheorem{claim}{Claim}

\newtheorem{lemma}{Lemma}
\newtheorem{example}{Example}
\newtheorem{theorem}{Theorem}

\newcommand{\mytodo}[2]{\xspace}
\newcommand{\myrevtodo}[2]{{%
		\let\marginpamarginnote
		\reversemarginpar
		\renewcommand{\baselinestretch}{0.8}%
		}}
\newcommand{\myinlinetodo}[2]{\todo[size=\small, color=#1!50!white, inline,
	caption={}]{#2}\xspace}
\newcommand{\registerAuthor}[3]{%
	\expandafter\newcommand\csname #2com\endcsname[1]{\mytodo{#3}{\textsc{#2}:
			##1}}%
	\expandafter\newcommand\csname
	#2revcom\endcsname[1]{\myrevtodo{#3}{\textsc{#2}: ##1}}%
	\expandafter\newcommand\csname
	#2inline\endcsname[1]{\myinlinetodo{#3}{\textsc{#2}: ##1}}%
	\expandafter\newcommand\csname
	#2inlineLater\endcsname[1]{\lv{\myinlinetodo{#3}{\textsc{#2}: ##1}}}%
}
\registerAuthor{Niclas Boehmer}{nb}{orange}
\registerAuthor{Klaus Heeger}{kh}{green}

\newcommand{\rk}{\operatorname{rk}}

\newcommand{\forrot}{prohibited}
\newcommand{\Pvalid}{forced-respecting}

\tikzstyle{vertex}=[draw, circle, fill, inner sep = 2pt]
\tikzstyle{squared-vertex}=[draw, fill, inner sep = 2pt]
\usetikzlibrary{arrows.meta}

\DeclareMathOperator{\single}{single}

\newenvironment{claimproof}{\begin{proof}[Proof of Claim.]}{\end{proof}}

		\newcommand{\reduced}{\text{red}}
		\Crefname{enumi}{Step}{Steps}
		\crefname{enumi}{step}{steps}
		
\usepackage{amsmath,amsfonts}

\allowdisplaybreaks
\usetikzlibrary{calc}
\tikzstyle{vertex}=[draw, circle, fill, inner sep = 2pt]
\tikzstyle{bedge}=[line width=1.8pt]
\tikzstyle{squared-vertex}=[draw, fill, inner sep = 2pt]
\usetikzlibrary{arrows.meta}

\usepackage{algorithm}
\usepackage[noend]{algpseudocode}
\algnewcommand\algorithmicinput{\textbf{Input:}}
\algnewcommand\algorithmicoutput{\textbf{Output:}}
\algnewcommand\Input{\item[\algorithmicinput]}
\algnewcommand\Output{\item[\algorithmicoutput]}
\algnewcommand\algorithmicgoto{\textbf{GoTo }}
\algnewcommand\GoTo{\item[\algorithmicgoto]}
\algnewcommand{\IIf}[1]{\State\algorithmicif\ #1\ \algorithmicthen}
\algnewcommand{\EndIIf}{\unskip\ \algorithmicend\ \algorithmicif}
\algnewcommand{\IfThenElse}[3]{
	\State \algorithmicif\ #1\ \algorithmicthen\ #2\ \algorithmicelse\ #3}

\usepackage{stackengine}
\stackMath
\newcommand\tsup[2][2]{%
	\def\useanchorwidth{T}%
	\ifnum#1>1%
	\stackon[-.5pt]{\tsup[\numexpr#1-1\relax]{#2}}{\scriptscriptstyle\sim}%
	\else%
	\stackon[.5pt]{#2}{\scriptscriptstyle\sim}%
	\fi%
}

\DeclareMathOperator{\pend}{
	\succ 
	\overset{\raise0.3em\hbox{\text{\scriptsize{(rest)}}}}{\ldots}}

\makeatletter
\providecommand*{\cupdot}{%
	\mathbin{%
		\mathpalette\@cupdot{}%
	}%
}
\newcommand*{\@cupdot}[2]{%
	\ooalign{%
		$\m@th#1\cup$\cr
		\hidewidth$\m@th#1\cdot$\hidewidth
	}%
}
\makeatother

\DeclareMathOperator{\Ac}{Ac}

\newcommand{\decprob}[3]{%
	\begin{center}%
		\begin{minipage}{0.9\linewidth}%
			\textsc{#1}\\
			\textbf{Input:} #2\\
			\textbf{Question:} #3
		\end{minipage}%
	\end{center}%
}

\newcommand{\ISRForcedForbidden}{\textsc{Adapt SR to Forced and Forbidden Pairs}\xspace}
\newcommand{\ISRtiesForcedForbidden}{\textsc{Adapt Weakly SR with Ties to Forced and Forbidden Pairs}\xspace}
\newcommand{\ISRStrongtiesForcedForbidden}{\textsc{Adapt Strongly SR with Ties to Forced and Forbidden Pairs}\xspace}
\newcommand{\ISMForcedForbidden}{\textsc{Adapt SM to Forced and Forbidden Pairs}\xspace}
\newcommand{\ISMtiesForcedForbidden}{\textsc{Adapt Weakly SM with Ties to Forced and Forbidden Pairs}\xspace}
\newcommand{\ISMStrongtiesForcedForbidden}{\textsc{Adapt Strongly SM with Ties to Forced and Forbidden Pairs}\xspace}

\renewcommand\thmcontinues[1]{Continued}

\usepackage[inline]{enumitem}

\usepackage[inline]{enumitem}

\begin{document}
\maketitle

\begin{abstract}
We introduce the problem of adapting a stable matching to forced and forbidden pairs. 
Specifically, given a stable matching $M_1$, a set~$Q$ of forced pairs, and a set $P$ of forbidden pairs, we want to find a stable matching that includes all pairs from $Q$, no pair from~$P$, and that is as close as possible to~$M_1$.
We study this problem in four classical stable matching settings: \textsc{Stable Roommates (with Ties)} and \textsc{Stable Marriage (with Ties)}.

As our main contribution, we employ the theory of rotations for \textsc{Stable Roommates} to develop a polynomial-time algorithm for adapting \textsc{Stable Roommates} matchings to forced pairs.
In contrast to this, we show that the same problem for forbidden pairs is NP-hard. 
However, our polynomial-time algorithm for the case of only forced pairs can be extended to a fixed-parameter tractable algorithm with respect to the number of forbidden pairs when both forced and forbidden pairs are present. 
Moreover, we also study the setting where preferences contain ties. 
Here, depending on the chosen stability criterion, we show either that our algorithmic results can be extended or that formerly tractable problems become intractable.
\end{abstract}

\section{Introduction}

Alice was recently hired as a tech lead and the company gave her the possibility to select her own team of software developers. 
After doing so, as it is a company-wide policy to use pair programming, Alice faces the problem of grouping her developers into pairs.
Because Alice is a fan of stable matchings, she organizes this by asking each software developer for his or her preferences over the other developers.
Subsequently, she computes and implements a stable matching (i.e., a matching where no two developers prefer each other to their assigned partner). 
Unfortunately, after a couple of weeks, Alice notices that Bob and Carol, who currently work together, like each other a little bit too much so that they spend most of their time not working productively.
Thus, she wants to assign both of them to a different partner. 
In contrast, Alice learns that Dan and Eve, who currently do not work together, have quite complementary skill sets. 
She believes that both of them would greatly benefit from working with each other.
Now, she faces the problem of finding a new stable matching that respects her wishes. 
However, as she observed that most pairs initially needed some time to find a joint way of working, she wants to minimize the number of new pairs, i.e., she wants the new matching to be as close as possible to the current one. 

More formally, the problem can be described as follows.
Alice is given a stable matching~$M_1$ of some agents with preferences over each other, a set $Q$ of forced pairs (those pairs need to be included in the new matching) and a set $P$ of forbidden pairs (none of these pairs is allowed to appear in the new matching), and she wants to find a new stable matching respecting the forced and forbidden pairs which is as close as possible to $M_1$.
We initiate the study of the decision variant of this problem, where we are additionally given an integer~$k$ and the symmetric difference between the old and the new matching shall be upper-bounded by~$k$,
in the following classical stable matching settings: 
\textsc{Stable Roommates} and its bipartite variant \textsc{Stable Marriage}, both combined with strict preferences or preferences containing ties.
We refer to the resulting problems as \textsc{Adapt Stable Roommates/Marriage [with Ties] to Forced and Forbidden Pairs}.\footnote{We consider two notions of stability if preferences contain ties, i.e., weak and strong stability. In weak stability, an agent pair $\{a,b\}$ is blocking a matching if both strictly prefer each other to their current partner, whereas in strong stability it is sufficient if $a$ strictly prefers $b$ to its partner and $b$ is indifferent between $a$ and its partner.}
For all six problems arising this way, we either present a polynomial-time algorithm or prove its NP-hardness.
Moreover, we provide a complete picture of the problems' parameterized computational complexity\footnote{Our results here are mostly along the parameterized complexity classes FPT and W[1]. A problem is fixed-parameter tractable (in FPT) with respect to some parameter $t$ if there is an algorithm solving every instance~$\mathcal{I}$ of the problem in $f(t) \cdot |\mathcal{I}|^{\mathcal{O}(1)}$~time for some computable function $f$. Under standard complexity theoretical assumptions, problems that are W[1]-hard for some parameter do not admit an FPT algorithm with respect to this parameter.} with respect to the problem-specific parameters $|P|$, $|Q|$, and $k$.

\paragraph{Related Work.}
Since the introduction of the \textsc{Stable Marriage} problem by Gale and Shapley \cite{DBLP:journals/tamm/GaleS13}, numerous facets of stable matching problems have been extensively studied in computer science and related areas (see, e.g., the monographs of Gusfield and Irving~\cite{DBLP:books/daglib/0066875}, Knuth \cite{Knuth76}, and Manlove \cite{DBLP:books/ws/Manlove13}). 
Our problem combines two previously studied aspects of stable matching problems: forced and forbidden pairs, and incremental algorithms. 

Dias et al.~\cite{DBLP:journals/tcs/DiasFFS03} initiated the study of stable matching problems with forced and forbidden pairs. 
The classical task here is to decide whether there is a stable matching including all forced pairs and no forbidden pair.\footnote{Note that our problem reduces to the classical problem associated with forced and forbidden pairs if we set the allowed distance between~$M_1$ and the matching to be found to infinity.}
While this problem can be solved in polynomial time if the preferences do not contain ties both in the roommates and marriage context~\cite{DBLP:journals/tcs/DiasFFS03,DBLP:journals/tcs/FleinerIM07}, the problem is NP-complete in the presence of ties for weak stability for marriage and roommates instances, even if there is only one forced and no forbidden pair~\cite{DBLP:journals/tcs/ManloveIIMM02} or one forbidden and no forced pair~\cite{DBLP:journals/disopt/CsehH20}.
Motivated by the straightforward observation that a stable matching including all given forced pairs and no forbidden pairs might not exist, Cseh and Manlove~\cite{DBLP:journals/disopt/CsehM16} studied the problem of finding a matching minimizing the number of ``violated constraints'' (where a violated constraint is either a blocking pair or a forced pair not contained in the matching or a forbidden pair contained in the matching).

Our problem also has a clear ``incremental'' dimension in the sense that we want to make as few changes as possible to a stable matching to achieve a certain goal. 
Many authors have studied such incremental problems in the context of various stable matching scenarios~\cite{DBLP:conf/icalp/BhattacharyaHHK15,DBLP:conf/mfcs/BoehmerHN22,uschanged,DBLP:conf/aaai/BredereckCKLN20,Feigenbaum17,DBLP:conf/fsttcs/GajulapalliLMV20,DBLP:conf/fsttcs/Gupta0R0Z20,DBLP:journals/algorithmica/MarxS10}.  
In the works of Bhattacharya et al.~\cite{DBLP:conf/icalp/BhattacharyaHHK15}, Boehmer et al. \cite{DBLP:conf/mfcs/BoehmerHN22,uschanged}, Bredereck et al. \cite{DBLP:conf/aaai/BredereckCKLN20}, Gajulapalli et al. \cite{DBLP:conf/fsttcs/GajulapalliLMV20}, and Feigenbaum et al. \cite{Feigenbaum17}, the focus lied on problems related to adapting matchings to change:
We are given a (stable) matching of agents, then some type of change occurs (e.g., some agents revise their preferences or some agents get added or deleted) and a new (stable) matching shall be found.
Here, as in our problem, it is often assumed that changing a pair in the matching is costly so the new matching should be as close as possible to the old one. 
As a second type of incremental problems, Marx and Schlotter~\cite{DBLP:journals/algorithmica/MarxS10} and Gupta et al. \cite{DBLP:conf/fsttcs/Gupta0R0Z20} analyzed the computational complexity of problems where one is given a stable matching $M$ and the task is to find a larger (almost) stable matching which is close to $M$.
On a more general note, this paper fits into the stream of works on incremental combinatorial problems \cite{DBLP:conf/icalp/BhattacharyaHHK15,DBLP:conf/atal/BoehmerN21,DBLP:journals/siamcomp/CharikarCFM04,DBLP:conf/icalp/EisenstatMS14,HMS21} where one aims at efficiently adapting solutions to changing inputs and requirements (in our case the requirement is that certain pairs are forbidden or forced), a core issue in modern algorithmics.

\begin{table*}
	\begin{center}
	\def\arraystretch{1.4}
	\resizebox{\textwidth}{!}{\begin{tabular}{ c|c|c|c|c }%
	& \textsc{SM}/\textsc{Strongly SM with Ties} & \textsc{Weakly SM/SR with Ties} & \textsc{SR}& \textsc{Strongly SR with Ties} \\  \hline
		Forced &  P (Pr. \ref{pr:SM-poly}) & \multirow{3}{*}{\parbox{5cm}{\centering NP-h. and W[1]-h. wrt. $k$+number of ties for one forced or forbidden pair (Pr.~\ref{pr:SM-tiesForc})}} & P (Th. \ref{thm:fpt-new}) & P (Th. \ref{thm:fpt-strong})  \\ \cline{1-2}\cline{4-5}
		Forbidden & P (Pr. \ref{pr:SM-poly}) & & NP-h. (Th. \ref{th:NP-h})  & NP-h. (Th. \ref{th:NP-h})\\ \cline{1-2}\cline{4-5} 
		Forced and Forbidden & P (Pr. \ref{pr:SM-poly}) & & FPT wrt. \#forbidden pairs in $M_1$ (Th. \ref{thm:fpt-new})  & FPT wrt. \#forbidden pairs in $M_1$ (Th. \ref{thm:fpt-strong}) 		  
	\end{tabular}}
\end{center}
\caption{Overview of our results. ``$k$'' denotes the allowed size of the symmetric difference between the old and new matching.
} \label{table:ov} 
\end{table*}

\paragraph{Our Contributions.}
We initiate the study of adapting stable matchings to forced and forbidden pairs.  
We consider this problem in six different settings and provide a complete dichotomy for the problems' (parameterized) computational complexity with respect to the problem-specific parameters $|P|$, $|Q|$, and $k$. See \Cref{table:ov} for an overview of our results.

In the first (short) part of the paper (\Cref{sec:SM}), we consider the bipartite marriage setting. 
We prove that adapting to forced and forbidden pairs is polynomial-time solvable for \textsc{Stable Marriages} without ties and in case of ties in combination with strong stability (\Cref{pr:SM-poly}). 
However, in case ties in the preferences are allowed and we are searching for weakly stable matchings, we obtain NP-hardness and W[1]-hardness with respect to the summed number of ties and the allowed difference $k$ between the old and the new matching (\Cref{pr:SM-tiesForc}). 
These hardness results hold even if there is only one forced and no forbidden pair or if there is only one forbidden and no forced pair.
As \textsc{Stable Roommates} generalizes \textsc{Stable Marriage}, these hardness results also hold for \textsc{Weakly Stable Roommates with Ties}.

In the second (main) part of the paper (\Cref{se:SR}), we focus on the \textsc{Stable Roommates} problem. 
Here, we first prove that in contrast to the bipartite setting, \textsc{Adapt Stable Roommates to Forced and Forbidden Pairs} is NP-hard, even if there are only forbidden pairs (\Cref{th:NP-h}). 
In contrast to this, the problem is fixed-parameter tractable with respect to the number of forbidden pairs (contained in the given matching; \Cref{thm:fpt-new}). In particular, if there are only forced pairs, then the problem is polynomial-time solvable.
To the best of our knowledge, this is the first problem which is tractable for forced but intractable for forbidden pairs.\footnote{Note that any problem involving only forced pairs can be reduced to a problem involving only forbidden pairs by setting for each forced pair~$\{a, b\}$ all pairs containing~$a$ except for~$\{a, b\}$ to be forbidden.}
The FPT-algorithm for adapting a \textsc{Stable Roommates} matching to forced and forbidden pairs is our main technical contribution.
Our algorithm relies on exploiting the structure of the rotation poset for \textsc{Stable Roommates} instances in a clever way: 
For this, we observe that for most pairs there is a necessary (and a prohibited) rotation that needs to be part of (cannot be part of) a set of rotations corresponding to a stable matching containing the pair. 
Using this, we can modify the set of rotations corresponding to the given matching to minimally change it to include all forced pairs. 
In fact, using some additional information, 
it is also possible to exclude forbidden pairs by modifying the  rotation set. 
Note that as each forbidden pair in $P$ requires a change in the matching $M_1$, this algorithm also constitutes a fixed-parameter tractable algorithm for the allowed difference $k$ between the old and the new matching.
Lastly, we describe how our algorithm can be modified to also work for the \textsc{Strongly SR with Ties} problem by exploiting the more intricate structure of the rotation poset for this problem using similar ideas as for \textsc{Stable Roommates} (\Cref{thm:fpt-strong}).

We defer the proofs (or their completions) of all results marked by (\appmark) to the appendix.

\section{Preliminaries}

In \textsc{Stable Roommates} (SR), we are given a set $A=\{a_1,\dots, a_{2n}\}$ of agents 
where each agent has a 
subset $\Ac(a)\subseteq A\setminus \{a\}$ of agents it finds \emph{acceptable}. 
We assume that acceptability is symmetric, i.e., $a\in \Ac(a')$ for some $a,a'\in A$ implies that $a'\in \Ac(a)$. 
Moreover, each agent~$a\in A$ has (strict) preferences $\succ_a$ over all agents it accepts, i.e., a total order over the agents $\Ac(a)$.
For agents~$a,a',a''\in A$, agent~$a$ \emph{prefers} $a'$ to $a''$ if $a'\succ_{a} a''$.

For a set $A$ of agents, we use $\binom{A}{2} $ to denote the 2-element subsets of~$A$; abusing notation, we will call these 2-element subsets \emph{pairs} although they are unordered.
A \emph{matching} $M$ is a set of pairs~$\{a,a'\}\in {A \choose 2}$ with $a \in \Ac (a')$ and $a' \in \Ac (a)$, where each
agent appears in at most one pair. 
An agent $a$ is \emph{matched} in some matching~$M$ if $M$ contains a pair containing $a$.
If $a$ is not matched in~$M$, then $a$ is \emph{unmatched}. 
A matching is \emph{complete} if all agents are matched. 
For an agent $a\in A$ and a matching $M$, we denote by~$M(a)$ the partner of~$a$ in $M$, i.e., $M(a)=a'$ if $\{a,a'\}\in M$.
For two matchings $M$ and $M'$ and an agent $a$ matched in both $M$ and $M'$, we say that $a$ prefers $M$ to~$M'$ if $a$ prefers $M(a)$ to $M'(a)$. 
An agent pair $\{a,a'\}\in {A \choose 2}$ \emph{blocks} a matching $M$ if 
\begin{enumerate*}[label=(\roman*)]
	\item $a\in \Ac(a')$ and $a'\in \Ac(a)$,
	\item $a$ is unmatched or $a$ prefers $a'$ to $M(a)$, and
	\item $a'$ is unmatched or $a'$ prefers $a$ to~$M(a')$.
\end{enumerate*}
A matching which is not blocked by any agent pair is called \emph{stable}. 
An agent pair $\{a,a'\}\in {A \choose 2}$ is a \emph{stable pair} if there is a stable matching $M$ with $\{a,a'\}\in M$. 
For two matchings $M$ and $M'$, we denote by~$M\triangle M'$ the set of pairs that only appear in one of $M$ and $M'$, i.e., $M \triangle M'=\{\{a,a'\}\mid \big(\{a,a'\}\in M \wedge \{a,a'\}\notin M' 
\big) \vee \big(\{a,a'\}\notin M \wedge \{a,a'\}\in M' \big)\}$. 
The main problem studied in this paper is the following:
\decprob{\textsc{Adapt SR to Forced and Forbidden Pairs}}{A 
	set $A$ of agents with strict preferences over each other,
	a stable 
	matching~$M_1$, a set of forced pairs $Q\subseteq {A \choose 2}$, a set of forbidden pairs $P\subseteq {A \choose 2}$, and 
	an integer~$k$.}{Is there a stable matching~$M_2$ with $Q\subseteq M_2$, $M_2\cap P = \emptyset$, and  $|M_1 \triangle 
	M_2| \le k$?}

In \textsc{SR with Ties}, a generalization of \textsc{SR}, each agent $a\in A$ has weak preferences $\succsim_a$ over all agents it accepts, i.e., $\succsim_a$ is a weak order over the agents $\Ac(a)$.
For agents $a,a',a''\in A$, agent $a$ \emph{weakly prefers} $a'$ to $a''$ if $a'\succsim_a a''$, agent~$a$ is \emph{indifferent} between $a'$ and~$a''$ (denoted as $a'\sim_a a''$) if both $a'\succsim_a a''$ and $a''\succsim_a a'$, and $a$ \emph{strictly prefers} $a'$ to $a''$ (denoted as $a'\succ_a a''$) if $a'\succsim_a a''$ but not $a''\succsim_a a'$. 
We distinguish two different types of stability in the presence of ties: 
Under weak/strong stability, an agent pair~$\{a,a'\}\in {A \choose 2}$ \emph{blocks} a matching $M$ if (i) $a\in \Ac(a')$ and $a'\in \Ac(a)$, (ii) $a$ is unmatched or $a$ strictly prefers $a'$ to $M(a)$ and (iii) $a'$ is unmatched or $a'$ strictly/weakly prefers $a$ to~$M(a')$.
The problems \textsc{Adapt Weakly/Strongly SR with Ties to Forced and Forbidden pairs} are defined analogous to
\textsc{Adapt SR to Forced and Forbidden Pairs}, where instead of strict preferences weak preferences are given and weak, respectively, strong stability is required.
	
In the bipartite variant of SR called \textsc{Stable Marriage} (SM), the agents are partitioned into two set $U$ and $W$. 
Following standard terminology, we call the elements from $U$ \emph{men} and the elements from $W$ \emph{women}. 
For each $m\in U$, we have $\Ac(m)\subseteq W$ and for each $w\in W$ we have $\Ac(w)\subseteq U$.
Consequently, agents from one side can only be matched to and form blocking pairs with agents from the other side. 
All other definitions from above still apply. 
The \textsc{Adapt (Strongly/Weakly) SM (with Ties) to Forced and Forbidden Pairs} problems are defined analogously to the respective variants for \textsc{SR} (the only difference being that the given instance is ``bipartite'', i.e., the set of agents can be split into two sets accepting only agents from the other set).

\section{Stable Marriage} \label{sec:SM}

In this section, we study the problem of adapting stable matchings to forced and forbidden pairs in the bipartite marriage setting. 

\subsection{(Strongly) Stable Marriage} \label{sub:SM-withoutTies}
We start by analyzing the case where agents' preferences are strict or when we are interested in strong stability in the presence of ties. 
We show that our problem is polynomial-time solvable in these settings by a simple reduction to the polynomial-time solvable \textsc{Weighted (Strongly) Stable Marriage (with Ties)} problem, where we are given an SM instance and a weight function on the pairs and the task is to compute a minimum-weight stable matching:

\begin{restatable}{proposition}{SMpoly}\label{pr:SM-poly}
	\ISMForcedForbidden and \ISMStrongtiesForcedForbidden are solvable in $\mathcal{O}(n \cdot m\log n)$~time. 
\end{restatable}
\begin{proof}
    Both problems can be solved using the same approach:
	We assume that $P\cap Q=\emptyset$, as otherwise we have a trivial no instance. 
	We define a weight function~$w$ as follows:
	For each forbidden pair~$e \in P$, we set $w(e ) := 3\cdot n$.
	For each forced pair~$e \in Q\setminus M_1$ that is not part of $M_1$, we set $w(e) : = 2 - 3 \cdot n$.
	For each forced pair~$e \in Q \cap M_1$ that is part of $M_1$, we set $w(e) := - 3\cdot n$.
	For each pair~$e \in M_1 \setminus (P \cup Q)$ that is part of $M_1$ but neither forced nor forbidden, we set $w (e) = 0$.
	For each remaining pair~$e$, we set $w(e) := 2$.
	We compute a minimum-weight stable matching~$M^*$ in $\mathcal{O} (n\cdot m \log n)$ time (see \cite{DBLP:journals/jcss/Feder92} for strict preferences and \cite{DBLP:conf/isaac/Kunysz18} for the case of ties with strong stability).
	Note that $w(M^*) = 3 \cdot n \cdot (|P \cap M^*| - |M^* \cap Q|) + 2 |M^* \setminus M_1| = 3 \cdot n (|P \cap M^*| - |M^* \cap Q|) + |M^* \triangle M_1| $ using that each stable matching has the same size by the Rural Hospitals Theorem \cite{Manlove99,roth1986allocation} (and thus $|M^*| = |M_1|$) for the second inequality.
	Since $|M^*| \le n$, it follows that $w(M^*) \le - 3 \cdot n \cdot |Q| + k$ if and only if $P \cap M^* = \emptyset$, $Q\subseteq M^*$, and $|M^* \triangle M_1 | \le k$.
\end{proof}

\subsection{Weakly Stable Marriage With Ties} \label{sub:SM-withties}

In contrast to the previous polynomial-time solvability result for strict preferences and for strong stability from \Cref{sub:SM-withoutTies}, we obtain strong intractability results if we consider weak stability.
Note that for \textsc{Weakly SM with Ties} already deciding the existence of a stable matching containing a single forced pair~\cite{DBLP:journals/tcs/ManloveIIMM02} or a single forbidden pair~\cite{DBLP:journals/disopt/CsehH20} is NP-complete, implying that \ISMtiesForcedForbidden\ is NP-complete already if $|P| = 1$ or~$|Q| = 1$.
We extend this hardness by showing W[1]-hardness when parameterized by the number of ties plus~$k$.

\begin{restatable}{proposition}{SMtiesForc}\label{pr:SM-tiesForc}
	\ISMtiesForcedForbidden restricted to instances where only agents from one side of the bipartition have ties in their preferences parameterized by the 
	number of ties  plus $k$ is W[1]-hard, even if~$|Q| = 
	1$ and $P = \emptyset$ or $Q = \emptyset $ and $|P| = 1$.
\end{restatable}

\begin{figure*}
	\begin{center}
		\begin{tikzpicture}[xscale = 2]
		\node[vertex, label=270:$u^*$] (a1) at (0,0) {};
		\node[vertex, label=270:$u_{\single}$] (a2) at (2,0) {};
		\node[vertex, label=270:$u_1$] (a3) at (3,0) {};
		\node[vertex, label=270:$u_2$] (a4) at (4,0) {};
		\node[vertex, label=270:$u_3$] (a5) at (5,0) {};
		
		\node[vertex, label=90:$w^*$] (b1) at (0,2) {};
		\node[vertex, label=90:$w_{\single}$] (b2) at (2,2) {};
		\node[vertex, label=90:$w_1$] (b3) at (3,2) {};
		\node[vertex, label=90:$w_2$] (b4) at (4,2) {};
		\node[vertex, label=90:$w_3$] (b5) at (5,2) {};
		
		\draw (a1) edge node[pos=0.2, fill=white, inner sep=2pt] {\scriptsize $2$}  node[pos=0.9, fill=white, inner sep=2pt] {\scriptsize $4$} (b3);
		\draw (a3) edge node[pos=0.1, fill=white, inner sep=2pt] {\scriptsize $3$}  node[pos=0.76, fill=white, inner sep=2pt] {\scriptsize $2$} (b1);
		
		\draw (a1) edge node[pos=0.2, fill=white, inner sep=2pt] {\scriptsize $3$}  node[pos=0.9, fill=white, inner sep=2pt] {\scriptsize $3$} (b4);
		\draw (a4) edge node[pos=0.1, fill=white, inner sep=2pt] {\scriptsize $4$}  node[pos=0.76, fill=white, inner sep=2pt] {\scriptsize $3$} (b1);
		
		\draw (a1) edge node[pos=0.2, fill=white, inner sep=2pt] {\scriptsize $4$}  node[pos=0.9, fill=white, inner sep=2pt] {\scriptsize $3$} (b5);
		\draw (a5) edge node[pos=0.1, fill=white, inner sep=2pt] {\scriptsize $3$}  node[pos=0.76, fill=white, inner sep=2pt] {\scriptsize $4$} (b1);
		
		\draw (a1) edge node[pos=0.2, fill=white, inner sep=2pt] {\scriptsize $5$}  node[pos=0.76, fill=white, inner sep=2pt] {\scriptsize $5$} (b1);
		\draw (a1) edge[ultra thick] node[pos=0.2, fill=white, inner sep=2pt] {\scriptsize $1$}  node[pos=0.9, fill=white, inner sep=2pt] {\scriptsize $2$} (b2);
		
		\draw (a2) edge[ultra thick] node[pos=0.1, fill=white, inner sep=2pt] {\scriptsize $2$}  node[pos=0.76, fill=white, inner sep=2pt] {\scriptsize $1$} (b1);
		\draw (a2) edge node[pos=0.1, fill=white, inner sep=2pt] {\scriptsize $1$}  node[pos=0.9, fill=white, inner sep=2pt] {\scriptsize $3$} (b3);
		
		\draw (a3) edge node[pos=0.1, fill=white, inner sep=2pt] {\scriptsize $1$}  node[pos=0.9, fill=white, inner sep=2pt] {\scriptsize $1$} (b2);
		\draw (a3) edge[ultra thick] node[pos=0.1, fill=white, inner sep=2pt] {\scriptsize $1$}  node[pos=0.9, fill=white, inner sep=2pt] {\scriptsize $1$} (b3);
		
		\draw (a4) edge node[pos=0.1, fill=white, inner sep=2pt] {\scriptsize $2$}  node[pos=0.9, fill=white, inner sep=2pt] {\scriptsize $2$} (b3);
		\draw (a4) edge[ultra thick] node[pos=0.28, fill=white, inner sep=2pt] {\scriptsize $3$}  node[pos=0.73, fill=white, inner sep=2pt] {\scriptsize $1$} (b4);
		\draw (a4) edge node[pos=0.25, fill=white, inner sep=2pt] {\scriptsize $1$}  node[pos=0.76, fill=white, inner sep=2pt] {\scriptsize $2$} (b5);
		
		\draw (a5) edge node[pos=0.2, fill=white, inner sep=2pt] {\scriptsize $2$}  node[pos=0.76, fill=white, inner sep=2pt] {\scriptsize $1$} (b4);
		\draw (a5) edge[ultra thick] node[pos=0.2, fill=white, inner sep=2pt] {\scriptsize $1$}  node[pos=0.76, fill=white, inner sep=2pt] {\scriptsize $1$} (b5);  
		
		\end{tikzpicture}
		
	\end{center}
	\caption{An example of the reduction from \Cref{pr:SM-tiesForc} for $P = \emptyset$ and $|Q| = 1$. Edges from~$M_1$ are depicted bold. The preferences of the agents are encoded in the numbers on the edges: For an edge $\{a,a'\}$, the number~$x$ closer to~$a$ denotes the position in which~$a'$ appear in the preference order of $a$, i.e., there are ${x -1}$~agents which $a$ strictly prefers to~$a'$.}
	\label{fig:forced:NP}
\end{figure*}
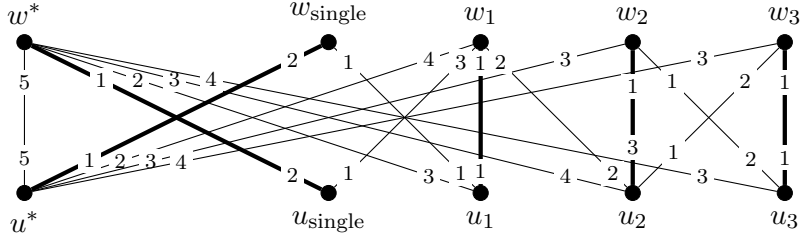
\begin{proof}
	\noindent \textbf{$\boldsymbol{P = \emptyset}$ and $\boldsymbol{|Q| = 1}$.}
	We reduce from the following problem related to 
	finding a complete stable matching in an \textsc{SM with Ties} 
	instance which we call \textsc{Local Search Complete Weakly Stable Marriage with Ties}: Given a 
	\textsc{SM with Ties} instance, and a stable matching~$N$ leaving only two agents unmatched, decide whether there exists a complete stable 
	matching~$N^*$ with $|N \triangle N^*| \le \ell$. Marx and 
	Schlotter showed that this problem 
	is W[1]-hard parameterized by the number ties plus~$\ell$, even if only the 
	preferences of agents from~$W$ contain ties \cite[Theorem 
	2]{DBLP:journals/algorithmica/MarxS10}.\footnote{Notably, 
		Marx and 
		Schlotter \cite{DBLP:journals/algorithmica/MarxS10} use a different measure 
		for the difference between two matchings, i.e., the number of agents that 
		are matched differently. However, as we here know that $|N|  + 1 = |N^*|$, this distance measure can be written as $|N\triangle N^*| + 1$.}
	
	We now establish a reduction from the above defined problem 
	to \ISMtiesForcedForbidden with only one forced and no forbidden pair.
	Let $(U \cup W,\mathcal{P})$ be an instance of \textsc{SM with Ties}, and 
	let $N$ be a stable matching of size $n-1$, where $n = |U| = |W|$. 
	Let $u_{\single}$ and $w_{\single}$ be the two agents unmatched in $N$.
	We add two agents~$u^*$ and~$w^*$. 
	The preferences of~$u^*$ respectively~$w^*$ start with all agents from $W$ respectively~$U$ in some arbitrary order followed by $w^*$ respectively~$u^*$. We set~$Q:= 
	\{\{u^*, w^*\}\}$ to be the set of forced pairs. Moreover, we add $u^*$ at 
	the 
	end of the preferences of every agent from~$W$ and $w^*$ at the end of the preferences 
	of every agent from~$U$. 
	Matching~$M_1$ is defined as $M_1 := N\cup \{\{u^*, w_{\single}\}, \{u_{\single}, 
	w^*\}\}$.
	The stability of~$M_1$ follows from the stability of $N$, as neither 
	$w^*$ nor $u^*$ is part of a blocking 
	pair.
	We set $k: = \ell + 3$. We now prove that there exists a complete stable matching~$N^*$ with $|N\triangle N^*| \leq k$ if and only if there exists a stable 
	matching $M^*$ such that $\{u^*, w^*\} \in 
	M^*$ and $| M_1 \triangle M^*| \le k$.
	Note that the constructed instance contains the same ties as $(U \cup W, \mathcal{P})$ (see \Cref{fig:forced:NP} for a visualization of the construction).
	
	$(\Rightarrow):$
	Given a complete stable matching $N^*$ with $|N \triangle N^*| \le \ell$, we 
	set~$M^* := N^* \cup \{\{u^*, 
	w^*\}\}$.
	Since $|N \triangle N^*| \le \ell$, it follows that $|M_1 \triangle M^*| \le 
	\ell+ 3 = k$.
	By definition, $M^* $ contains~$Q = \{\{u^*, w^*\}\}$.
	It remains to show that $M^*$ is a stable matching.
	Because $N^*$ is a matching, $M^*$ is also a matching.
	As $N^*$ is stable, every blocking pair must contain~$w^*$ or~$u^*$.
	Since $N$ is a complete matching and as all agents from $U \cup W$ rank $u^*$ and $w^*$ last, no agent from~$U \cup W$ prefers $u^*$ or~$w^*$ to their partner in $M^*$.
	Therefore, $M^*$ is stable. 
	
	$(\Leftarrow):$   
	Vice versa, let $M^*$ be a stable matching 
	such that $\{u^*, w^*\} \in M^*$ and $| M_1 \triangle M^*| \le k$.
	As $M^*$ contains $\{u^*,w^*\}$, it follows that for each agent~$a\in U \cup W$, agent~$a$ has to be matched to an agent it prefers to $u^*$ and $w^*$, i.e., an agent from $U \cup W$.
	Therefore, $N := M^* \setminus \{ \{u^*, w^*\}\}$ is a complete matching on~$U \cup W$.
	It is also a stable one, as any blocking pair would also be a blocking pair for $M^*$.
	Furthermore, $|N \triangle N^*| =|M_1 \triangle M^*|  - 3 \le \ell$. \medskip
	
	\noindent\textbf{$\boldsymbol{|P| = 1}$ and $\boldsymbol{Q = \emptyset}$.}
	We now modify our reduction to show hardness when $Q = \emptyset$ and $|P| = 1$ using a somewhat similar approach as in the NP-hardness of \textsc{Stable Marriage with Ties and Forced Edges} by Cseh and Heeger~\cite{DBLP:journals/disopt/CsehH20}.
	We again reduce from the above described \textsc{Local Search Complete Weakly Stable Marriage with Ties} problem and modify an instance of this problem as follows.
	This time, we only add~$w^*$ (but not~$u^*$) and additionally add two agents~$u'$ and~$w'$.
	Again, $w^*$ is added at the end of the preferences of each agent from~$U$.
	The preferences of~$w^*$ start with all agents from~$U$ (in an arbitrary order), followed by~$u'$.
	Agent~$u'$ prefers~$w^*$ to~$w'$, while $w'$ only accepts~$u'$.
	Finally, we set~$M_1 := N \cup \{\{u_{\single}, w^*\}, \{u', w'\}\}$, the set of forbidden pairs to $P := \{ \{u', w'\}\}$, and $k := \ell + 3$.
	
	$(\Rightarrow)$:
	Given a complete stable matching $N^*$ with $|N \triangle N^*| \le \ell$, we 
	set~$M^* := N^* \cup \{\{m' , w^*
	\}\}$.
	Since $|N \triangle N^*| \le \ell$, it follows that $|M_1 \triangle M^*| \le 
	\ell+ 3 = k$.
	By definition, $M^* $ does not contain the forbidden pair~$\{u', w'\}$.
	It remains to show that $M^*$ is a stable matching.
	Because $N^*$ is a matching, also $M^*$ is a matching.
	As $N^*$ is stable, every blocking pair must contain~$w^*$, $u'$, or $w'$.
	Since $N$ is a complete matching, no agent from~$U \cup W$ prefers to be matched to~$w^*$, so $w^*$ is not part of a blocking pair.
	Neither~$u'$ nor $w'$ are part of a blocking pair as $u'$ prefers~$w^*$ to~$w'$.
	Therefore, $M^*$ is stable.
	
	$(\Leftarrow):$   
	Vice versa, let $M^*$ be a stable matching 
	such that $\{u', w'\} \notin M^*$ and $| M_1 \triangle M^*| \le k$.
	As $\{u', w'\} \notin M^*$, it follows that $\{u', w^*\} \in M^*$ (otherwise $\{u', w'\}$ would block~$M^*$).
	As $\{u, w^*\}$ does not block~$M^*$ for some~$u \in U$, it follows that each $u\in U$ is matched to some~$w \in W$.
	As $|U|=|W|$ it follows that $N := M^* \setminus \{ \{u', w^*\}\}$ is a complete matching on~$U \cup W$. 
	It is also a stable one, as any blocking pair would also be a blocking pair for~$M^*$.
	Furthermore, $|N \triangle N^*| =|M_1 \triangle M^*|  - 3 \le \ell$. 
\end{proof}

\section{Stable Roommates} \label{se:SR}
As the strong intractability results for \ISMtiesForcedForbidden from \Cref{sub:SM-withties} extend to \ISRtiesForcedForbidden (as \textsc{SR with Ties} generalizes \textsc{SM with Ties}) in this section we focus on \textsc{Adapt (Strongly) SR (with Ties) to Forced and Forbidden Pairs}.
We first prove in \Cref{SR:hard} that adapting a \textsc{SR} matching to forbidden pairs is NP-hard even without ties. 
Afterwards, in our core \Cref{sub:FPT}, we prove that \ISRForcedForbidden parameterized by the number of forbidden pairs that appear in $M_1$ is fixed-parameter tractable (and thus that adapting an \textsc{SR} matching to forced pairs is polynomial-time solvable) by exploiting the rotation poset.
Lastly, in \Cref{fpt:strong}, we extend this result to also work for  \ISRStrongtiesForcedForbidden. 

\subsection{NP-hardness of Adapt SR to Forbidden Pairs} \label{SR:hard}
In this section, we prove that in contrast to the bipartite marriage setting, \ISRForcedForbidden (without ties) is already NP-hard (even if we only have forbidden pairs). 

\begin{theorem} \label{th:NP-h}
	\ISRForcedForbidden is NP-hard, even if $Q = \emptyset$ and $P\subseteq M_1$. 
\end{theorem}
	\begin{figure*}
		\centering
		\begin{subfigure}[t]{0.3\textwidth}
			\centering
			\resizebox{.65\textwidth}{!}{\begin{tikzpicture}
				\node[vertex, label=180:$a_1^v$] (a1) at (0,0) {};
				\node[vertex, label=180:$a_2^v$] (a2) at (0,2) {};
				\node[vertex, label=180:$a_3^v$] (a3) at (0,3) {};
				\node[vertex, label=180:$a_4^v$] (a4) at (0,4) {};
				\node[vertex, label=180:$a_5^v$] (a5) at (0,5) {};
				
				\node[vertex, label=0:$b_1^v$] (b1) at (2,0) {};
				\node[vertex, label=0:$b_2^v$] (b2) at (2,2) {};
				\node[vertex, label=0:$b_3^v$] (b3) at (2,3) {};
				\node[vertex, label=0:$b_4^v$] (b4) at (2,4) {};
				\node[vertex, label=0:$b_5^v$] (b5) at (2,5) {};
				
				\draw[line width=2.5pt] (a1) edge node[pos=0.2, fill=white, inner sep=2pt] {\scriptsize $1$}  node[pos=0.76, fill=white, inner sep=2pt] {\scriptsize $2$} (b1);
				\draw (a1) edge node[pos=0.2, fill=white, inner sep=2pt] {\scriptsize $2$}  node[pos=0.76, fill=white, inner sep=2pt] {\scriptsize $1$} (b2);
				
				\draw (a2) edge node[pos=0.2, fill=white, inner sep=2pt] {\scriptsize $|N(v)|+3$}  node[pos=0.76, fill=white, inner sep=2pt] {\scriptsize $1$} (b1);
				\draw[line width=2.5pt] (a2) edge node[pos=0.15, fill=white, inner sep=2pt] {\scriptsize $2$}  node[pos=0.76, fill=white, inner sep=2pt] {\scriptsize $2$} (b2);
				\draw (a2) edge node[pos=0.25, fill=white, inner sep=2pt] {\scriptsize $1$}  node[pos=0.76, fill=white, inner sep=2pt] {\scriptsize $3$} (b3);
				
				\draw (a3) edge node[pos=0.25, fill=white, inner sep=2pt] {\scriptsize $1$}  node[pos=0.76, fill=white, inner sep=2pt] {\scriptsize $3$} (b2);
				\draw[line width=2.5pt] (a3) edge node[pos=0.2, fill=white, inner sep=2pt] {\scriptsize $2$}  node[pos=0.76, fill=white, inner sep=2pt] {\scriptsize $1$} (b3);
				
				\draw (a4) edge node[pos=0.25, fill=white, inner sep=2pt] {\scriptsize $2$}  node[pos=0.76, fill=white, inner sep=2pt] {\scriptsize $2$} (b3);
				\draw[line width=2.5pt] (a4) edge node[pos=0.2, fill=white, inner sep=2pt] {\scriptsize $3$}  node[pos=0.76, fill=white, inner sep=2pt] {\scriptsize $1$} (b4);
				\draw (a4) edge node[pos=0.25, fill=white, inner sep=2pt] {\scriptsize $1$}  node[pos=0.76, fill=white, inner sep=2pt] {\scriptsize $2$} (b5);
				
				\draw (a5) edge node[pos=0.2, fill=white, inner sep=2pt] {\scriptsize $1$}  node[pos=0.76, fill=white, inner sep=2pt] {\scriptsize $2$} (b4);
				\draw[line width=2.5pt] (a5) edge node[pos=0.3, fill=white, inner sep=2pt] {\scriptsize $2$}  node[pos=0.7, fill=white, inner sep=2pt] {\scriptsize $1$} (b5);    
				\end{tikzpicture}}
			\caption{Initial matching.}
			\label{fig:y equals x}
		\end{subfigure}
		\hfill
		\begin{subfigure}[t]{0.3\textwidth}
			\centering
			\resizebox{.65\textwidth}{!}{\begin{tikzpicture}
				\node[vertex, label=180:$a_1^v$] (a1) at (0,0) {};
				\node[vertex, label=180:$a_2^v$] (a2) at (0,2) {};
				\node[vertex, label=180:$a_3^v$] (a3) at (0,3) {};
				\node[vertex, label=180:$a_4^v$] (a4) at (0,4) {};
				\node[vertex, label=180:$a_5^v$] (a5) at (0,5) {};
				
				\node[vertex, label=0:$b_1^v$] (b1) at (2,0) {};
				\node[vertex, label=0:$b_2^v$] (b2) at (2,2) {};
				\node[vertex, label=0:$b_3^v$] (b3) at (2,3) {};
				\node[vertex, label=0:$b_4^v$] (b4) at (2,4) {};
				\node[vertex, label=0:$b_5^v$] (b5) at (2,5) {};
				
				\draw (a1) edge node[pos=0.2, fill=white, inner sep=2pt] {\scriptsize $1$}  node[pos=0.76, fill=white, inner sep=2pt] {\scriptsize $2$} (b1);
				\draw[line width=2.5pt] (a1) edge node[pos=0.2, fill=white, inner sep=2pt] {\scriptsize $2$}  node[pos=0.76, fill=white, inner sep=2pt] {\scriptsize $1$} (b2);
				
				\draw[line width=2.5pt] (a2) edge node[pos=0.2, fill=white, inner sep=2pt] {\scriptsize $|N(v)|+3$}  node[pos=0.76, fill=white, inner sep=2pt] {\scriptsize $1$} (b1);
				\draw (a2) edge node[pos=0.15, fill=white, inner sep=2pt] {\scriptsize $2$}  node[pos=0.76, fill=white, inner sep=2pt] {\scriptsize $2$} (b2);
				\draw (a2) edge node[pos=0.25, fill=white, inner sep=2pt] {\scriptsize $1$}  node[pos=0.76, fill=white, inner sep=2pt] {\scriptsize $3$} (b3);
				
				\draw (a3) edge node[pos=0.25, fill=white, inner sep=2pt] {\scriptsize $1$}  node[pos=0.76, fill=white, inner sep=2pt] {\scriptsize $3$} (b2);
				\draw[line width=2.5pt] (a3) edge node[pos=0.2, fill=white, inner sep=2pt] {\scriptsize $2$}  node[pos=0.76, fill=white, inner sep=2pt] {\scriptsize $1$} (b3);
				
				\draw (a4) edge node[pos=0.25, fill=white, inner sep=2pt] {\scriptsize $2$}  node[pos=0.76, fill=white, inner sep=2pt] {\scriptsize $2$} (b3);
				\draw[line width=2.5pt] (a4) edge node[pos=0.2, fill=white, inner sep=2pt] {\scriptsize $3$}  node[pos=0.76, fill=white, inner sep=2pt] {\scriptsize $1$} (b4);
				\draw (a4) edge node[pos=0.25, fill=white, inner sep=2pt] {\scriptsize $1$}  node[pos=0.76, fill=white, inner sep=2pt] {\scriptsize $2$} (b5);
				
				\draw (a5) edge node[pos=0.2, fill=white, inner sep=2pt] {\scriptsize $1$}  node[pos=0.76, fill=white, inner sep=2pt] {\scriptsize $2$} (b4);
				\draw[line width=2.5pt] (a5) edge node[pos=0.3, fill=white, inner sep=2pt] {\scriptsize $2$}  node[pos=0.7, fill=white, inner sep=2pt] {\scriptsize $1$} (b5);    
				\end{tikzpicture}}
			\caption{Matching $M^v$ ($v$ is selected to be part of the independent set).}
			\label{fig:three sin x}
		\end{subfigure}
		\hfill
		\begin{subfigure}[t]{0.3\textwidth}
			\centering
			\resizebox{.65\textwidth}{!}{\begin{tikzpicture}
				\node[vertex, label=180:$a_1^v$] (a1) at (0,0) {};
				\node[vertex, label=180:$a_2^v$] (a2) at (0,2) {};
				\node[vertex, label=180:$a_3^v$] (a3) at (0,3) {};
				\node[vertex, label=180:$a_4^v$] (a4) at (0,4) {};
				\node[vertex, label=180:$a_5^v$] (a5) at (0,5) {};
				
				\node[vertex, label=0:$b_1^v$] (b1) at (2,0) {};
				\node[vertex, label=0:$b_2^v$] (b2) at (2,2) {};
				\node[vertex, label=0:$b_3^v$] (b3) at (2,3) {};
				\node[vertex, label=0:$b_4^v$] (b4) at (2,4) {};
				\node[vertex, label=0:$b_5^v$] (b5) at (2,5) {};
				
				\draw[line width=2.5pt] (a1) edge node[pos=0.2, fill=white, inner sep=2pt] {\scriptsize $1$}  node[pos=0.76, fill=white, inner sep=2pt] {\scriptsize $2$} (b1);
				\draw (a1) edge node[pos=0.2, fill=white, inner sep=2pt] {\scriptsize $2$}  node[pos=0.76, fill=white, inner sep=2pt] {\scriptsize $1$} (b2);
				
				\draw (a2) edge node[pos=0.2, fill=white, inner sep=2pt] {\scriptsize $|N(v)|+3$}  node[pos=0.76, fill=white, inner sep=2pt] {\scriptsize $1$} (b1);
				\draw (a2) edge node[pos=0.15, fill=white, inner sep=2pt] {\scriptsize $2$}  node[pos=0.76, fill=white, inner sep=2pt] {\scriptsize $2$} (b2);
				\draw[line width=2.5pt] (a2) edge node[pos=0.25, fill=white, inner sep=2pt] {\scriptsize $1$}  node[pos=0.76, fill=white, inner sep=2pt] {\scriptsize $3$} (b3);
				
				\draw[line width=2.5pt] (a3) edge node[pos=0.25, fill=white, inner sep=2pt] {\scriptsize $1$}  node[pos=0.76, fill=white, inner sep=2pt] {\scriptsize $3$} (b2);
				\draw (a3) edge node[pos=0.2, fill=white, inner sep=2pt] {\scriptsize $2$}  node[pos=0.76, fill=white, inner sep=2pt] {\scriptsize $1$} (b3);
				
				\draw (a4) edge node[pos=0.25, fill=white, inner sep=2pt] {\scriptsize $2$}  node[pos=0.76, fill=white, inner sep=2pt] {\scriptsize $2$} (b3);
				\draw (a4) edge node[pos=0.2, fill=white, inner sep=2pt] {\scriptsize $3$}  node[pos=0.76, fill=white, inner sep=2pt] {\scriptsize $1$} (b4);
				\draw[line width=2.5pt] (a4) edge node[pos=0.25, fill=white, inner sep=2pt] {\scriptsize $1$}  node[pos=0.76, fill=white, inner sep=2pt] {\scriptsize $2$} (b5);
				
				\draw[line width=2.5pt] (a5) edge node[pos=0.2, fill=white, inner sep=2pt] {\scriptsize $1$}  node[pos=0.76, fill=white, inner sep=2pt] {\scriptsize $2$} (b4);
				\draw (a5) edge node[pos=0.3, fill=white, inner sep=2pt] {\scriptsize $2$}  node[pos=0.7, fill=white, inner sep=2pt] {\scriptsize $1$} (b5);    
				\end{tikzpicture}}
			\caption{Matching $\bar M^v$ ($v$ is not selected to be part of the independent set).}
			\label{fig:five over x}
		\end{subfigure}
		\caption{The preferences of agents~$a_i^v$ and $b_i^v$ for some~$v\in V(G)$ with different matchings highlighted in bold. $\{a_2^v,b_2^v\}$ is the forbidden pair. The preferences of the agents are encoded in the numbers on the edges: For an edge $\{a,a'\}$, the number~$x$ closer to~$a$ denotes the position in which~$a'$ appears in the preferences of $a$, i.e., $a$ prefers exactly $x-1$ agents to~$a'$.}
		\label{fig:adaptSR}
	\end{figure*}
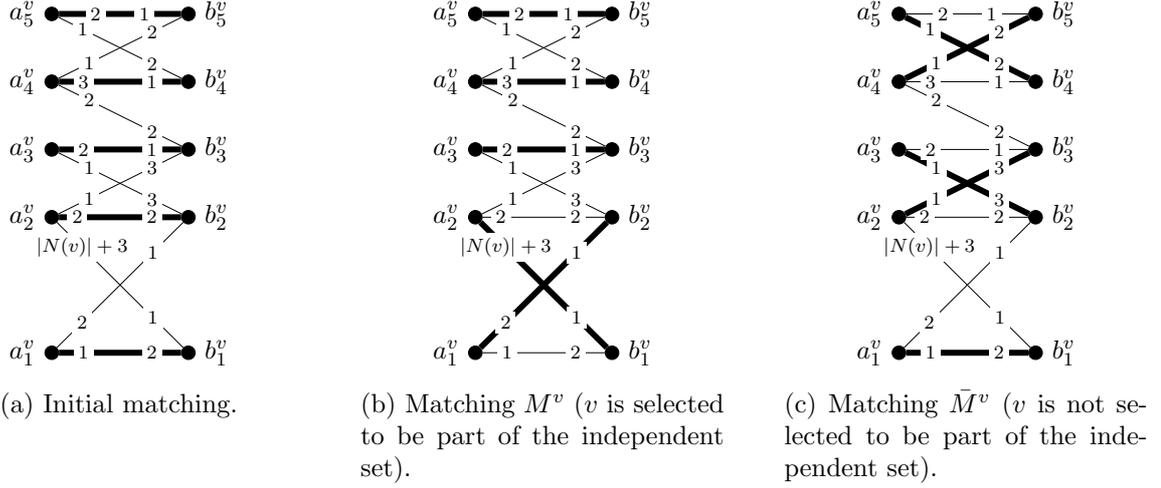
\begin{proof}
	We reduce from the NP-hard \textsc{Independent Set} problem~\cite{DBLP:conf/coco/Karp72}.
	Let $(G, \ell)$ be an instance of \textsc{Independent Set}.
	For a vertex~$v\in V(G)$, we denote by~$N(v)$ the set of its neighbors in~$G$.
	For each vertex~$v \in V(G)$, the \ISRForcedForbidden instance contains ten agents $a_1^v$, \dots, $a_5^v$, $b_1^v$, \dots, $b_5^v$.
	For each $v\in V(G)$, fix an arbitrary strict order of $\{a_2^w  \mid w \in N(v)\}$ and denote this order by $[N^* (v)]$.
	For each $v\in V(G)$ the preferences of the respective ten agents are as follows (see also \Cref{fig:adaptSR}): 
		\begin{align*}
		&a_1^v : b_1^v \succ b_2^v, \quad a_2^v : b_3^v \succ b_2^v \succ [N^*(v)] \succ b_1^v,\\ &a_3^v  : b_2^v \succ b_3^v, \quad a_4^v : b_5^v \succ b_3^v \succ b_4^v, \quad a_5^v : b_4^v \succ b_5^v \\
		&b_1^v  : a_2^v \succ a_1^v, \quad b_2^v  : a_1^v \succ a_2^v \succ a_3^v, \quad  b_3^v  : a_3^v \succ a_4^v \succ a_2^v, \quad\\ &b_4^v  : a_4^v \succ a_5^v, \quad b_5^v  : a_5^v \succ a_4^v
		\end{align*}
	Finally, we set $M_1 := \{\{a_i^v, b_i^v\} \mid i \in [5] , v \in V(G)\}$, $P := \{\{a_2^v, b_2^v\} \mid v \in V(G)\}$, and~$k:= 8 |V(G)| - 4\ell$.
	Note that $M_1$ is stable, as for each $v\in V(G)$, agents~$b_3^v$, $b_4^v$, $b_5^v$, and $a_1^v$ are matched to their top-choices (so they cannot be part of a blocking pair) and $a_2^v$ and $b_2^v$ are matched to their most preferred agents that are not listed above.
	
	$(\Rightarrow)$:
	Let~$X$ be an independent set of size~$\ell$ in $G$.
	For a vertex~$v \in V(G)$, we set~$M^v := \{ \{a_1^v, b_2^v\},\{a_2^v, b_1^v\},\{a_3^v, b_3^v\},\{a_4^v, b_4^v\},\{a_5^v, b_5^v\}\}$ and $\bar M^v := \{ \{a_1^v, b_1^v\},\{a_2^v, b_3^v\},\{a_3^v, b_2^v\},\{a_4^v, b_5^v\},\{a_5^v, b_4^v\}\}$.
	We set~$M^* := \bigcup_{v \in X} M^v \cup \bigcup_{v \in V \setminus X } \bar M^v$.
	Then $M^* \triangle M_1 =  \{ \{a_1^v, b_2^v\},\{a_2^v, b_1^v\},\allowbreak\{a_1^v, b_1^v\},\allowbreak\{a_2^v, b_2^v\} \mid v\in X\} \cup \{\{a_2^v, b_3^v\},\allowbreak\{a_3^v, b_2^v\},\allowbreak\{a_4^v, b_5^v\}, \allowbreak\{a_5^v, b_4^v\}, \allowbreak\{a_2^v, b_2^v\},\allowbreak \{a_3^v, b_3^v\},\allowbreak \{a_4^v, b_4^v\},\allowbreak \{a_5^v, b_5^v\} \mid v\in V(G) \setminus X\}$.
	Consequently, we have~$|M^* \triangle M_1| = 4|X| + 8 \cdot (|V(G)| -|X|) = 8|V(G) | - 4\ell$.
	As $M^*$ clearly does not contain any forbidden pair, it remains to show that $M^*$ is stable.
	
	It is straightforward to verify that no pair~$\{a_i^v, b_j^v\}$ for~${i, j \in [5]}$ and $v \in V(G)$ is blocking.
	The remaining acceptable pairs are~$\{a_2^v, a_2^w\}$ for some $\{v, w\}\in E(G)$.
	Since~$X$ is an independent set, we may assume without loss of generality that $v \notin X$.
	This implies that $M^* (a_2^v) = b_3^v \succ_{a_2^v} a_2^w$, implying that $\{a_2^v, a_2^w\}$ does not block~$M^*$.
	Thus, $M^*$ is stable.
	
	$(\Leftarrow)$:
	Let~$M^*$ be a stable matching with~$|M^* \triangle M_1| \le k = 8|V(G)| - 4\ell$ in the constructed instance.
	First note that the Rural Hospitals Theorem~\cite[Theorem~4.5.2]{DBLP:books/daglib/0066875} (which states that every stable matching matches the same set of agents) implies that every stable matching is complete in the constructed instance.
	Consequently, $M^*$ does not contain a pair of the form~$\{a_2^v, a_2^w\}$ (as otherwise one of $b_1^v$, \dots, $b_5^v$ would be unmatched in~$M^*$).
	Thus, for each~$v \in V(G)$, we have $M^* (a_2^v) \in \{b_1^v, b_3^v\}$ (recall that we forbid the pair $\{a_2^v,b_2^v\}$ for all $v\in V(G)$).
	Note that $X := \{v \in V(G) \mid \{a_2^v, b_1^v\} \in M^*\}$ is an independent set:
	If $\{v, w\} \in E(G)$ for $v\neq w \in X$, then $\{a_2^v, a_2^w\}$ blocks~$M^*$.
	
	It remains to show that $|X| \ge \ell$.
	For each $v\in X$, we have $\{a_2^v, b_1^v\} \in M^*$ (by the definition of~$X$) and $\{a_1^v, b_2^v\} \in M^*$ (as $a_1^v$ would be unmatched otherwise).
	Consequently, $|\big(M^* \triangle M_1 \big) \cap \{\{a_i^v, b_j^v\}: i, j\in [5]\}| \ge 4$.
	For each $v \in V(G) \setminus X$, by the definition of $X$ we have $\{a_2^v,b_3^v\}\in M^*$. Moreover, note that $M^*$ contains~$\{a_1^v, b_1^v\}$ (as otherwise $b_1^v$ would be unmatched) and $\{a_3^v, b_2^v\}$ (otherwise $b_2^v$ would be unmatched).
	Further, $M^*$ contains~$\{a_4^v, b_5^v\}$ (otherwise $\{a_4^v, b_3^v\}$ would be blocking) and $\{a_5^v, b_4^v\}$ (otherwise $a_5^v$ and $b_4^v$ would be unmatched).
	Consequently, we have $|\big(M^* \triangle M_1 \big) \cap \{\{a_i^v, b_j^v\}: i, j\in [5]\}| \ge 8$.
	Summing up, we get that $k = 8|V(G)| - 4 \ell \ge |M_1 \triangle M^* | \ge 4|X| + 8 (|V(G) | - |X|) = 8 |V(G)| - 4|X|$, which is equivalent to $|X| \ge \ell$.
\end{proof}

\subsection{(FPT-)Algorithm for Adapt SR to Forced and Forbidden Pairs} \label{sub:FPT}
In this section, we develop an FPT-algorithm for the \ISRForcedForbidden problem parameterized by the number of forbidden pairs in $M_1$ (note that this algorithm is a polynomial-time algorithm if no forbidden pairs are present).
Our algorithm  heavily relies on the rotation poset for \textsc{Stable Roommates}.
Thus, we start this section by defining rotations (\Cref{rotations:intro}) and describing the high-level idea of our algorithm together with proving some useful facts concerning rotations (\Cref{rotations:facts}), before we present our algorithm (\Cref{alg}).

In the following we assume that all considered stable matchings (and in particular the initial matching $M_1$) are complete matchings, as we can otherwise modify the instance accordingly in $\mathcal{O}(m)$ time.\footnote{If $M_1$ is not complete, let $B$ be the set of agents unmatched in $M_1$. For each agent~$b\in B$, we add an agent $b'$ to the instance which only finds $b$ acceptable and which is added at the end of the preferences of $b$. Then, using the Rural Hospitals Theorem for SR~\cite{DBLP:books/daglib/0066875}, which states that each stable matching in a SR instance matches the same set of agents, it follows that all stable matchings in the modified instance contain pairs~$\{\{b,b'\}\mid b\in B\}$. Consequently, the modified instance is equivalent to the original one.} 
\subsubsection{Rotations: Introduction} \label{rotations:intro}
We first formally define what a rotation is, then discuss their relationship to Irving's algorithm, and lastly identify different types of rotations.

\paragraph{Basic Definitions for Rotations.}
	For an instance of \textsc{SR}, an \emph{exposed rotation} is a sequence of agent pairs~$(a_{i_0}, a_{j_0}), \dots, (a_{i_{r-1}}, a_{j_{r-1}})$ such that, for each~${s\in [r]}$, agent~$a_{i_s}$ ranks $a_{j_{s}}$ first and $a_{j_{s+1}}$ second (where all indices in this paragraph are taken modulo $r$).\footnote{
	Notably, a rotation has no fixed start point, as we can start with any pair from the sequence resulting in shifted versions of the same rotation. In the following, we do not distinguished between these different shifted variants of the same rotation as they are the same for our purposes.}
	\emph{Eliminating} an exposed rotation $(a_{i_0}, a_{j_0}), \dots, (a_{i_{r-1}}, a_{j_{r-1}})$ means deleting, for all~${s \in [r]}$, all agents which~$a_{j_s}$ ranks after $a_{i_{s-1}}$ from the preferences of~$a_{j_s}$. 
	The \emph{dual} $\bar{\varphi}$ of a rotation~$\varphi=(a_{i_0}, a_{j_0}), \dots, (a_{i_{r-1}}, a_{j_{r-1}})$ is $\bar{\varphi}=(a_{j_0}, a_{i_{r-1}}), (a_{j_1}, a_{i_0}), (a_{j_2}, a_{i_1}), \dots, (a_{j_{r-1}}, a_{j_{r-2}})$.
	Note that the dual of the dual of a rotation is again the rotation itself.
	
\paragraph{Irving's Algorithm.}	
	The theory of rotations is closely connected to Irving's algorithm \cite{DBLP:journals/jal/Irving85}. 
	Irving's algorithm constructs a stable matching in an \textsc{SR} instance (if it exists) in two phases. 
	In the first phase, similar to the Gale-Shapely algorithm for \textsc{SM}, agents make proposals to each other, which are accepted or rejected. 
	Doing so, certain parts of the agent's preferences get deleted. 
	Let $P_0$ be the preference profile of the agents after the termination of Phase 1. 
	Now Phase~2 consists of eliminating exposed rotations one after each other until no rotation is exposed anymore (note that after eliminating a rotation, some agents delete agents from their preferences, causing the set of exposed rotations to change). 
	If no rotation is exposed, then either there is at least one agent with empty preferences, implying that no stable matching exists, or every agent has exactly one other agent left in its preferences, implying that matching the agents to the remaining agent in their preferences results in a stable matching.
	We call a preference profile a \emph{stable table} if it can be derived from $P_0$ after successively eliminating exposed rotations.
	For an instance of \textsc{Stable Roommates}, the \emph{rotations} are sequences of agent pairs which may arise as an exposed rotation in some execution of Irving's algorithm (since Irving's algorithm may eliminate any exposed rotation, different executions of Irving's algorithm may result in different stable matchings and different exposed rotations)
	
	\paragraph{(Non)-Singular Rotations and Further Definitions.}	
	Using the view of Irving's algorithm now allows us to identify different types of rotations. 
	A rotation $\varphi$ is \emph{nonsingular} if its dual $\bar{\varphi}$ is again a rotation.
	Otherwise, the rotation $\varphi$ is \emph{singular}.
	For two rotations~$\varphi, \rho$, we say $\varphi \vartriangleright \rho$ (or $\varphi$ precedes~$\rho$) if $\varphi $ must be eliminated from a stable table to give a stable table where $\rho $ is exposed. 
	A set~$Z$ of rotations is \emph{closed} if whenever $\rho \in Z$ and $\varphi \vartriangleright \rho$, then also $\varphi \in Z$.
	A set~$Z$ of rotations is \emph{complete} if it contains all singular rotation and for each nonsingular rotation~$\varphi$, it contains either $\varphi$ or $\bar \varphi$.
	An agent pair is called \emph{fixed} if it is contained in every stable matching.
	An agent~$b$ is a \emph{stable partner} of agent~$a$ if there is a stable matching containing~$\{a, b\}$, i.e., if $\{a, b\}$ is a stable pair.
	
	\begin{example}[label=ex:rotation]
	 Consider the following instance of \textsc{Stable Roommates} (in fact, this is even an instance of \textsc{Stable Marriage}).
	 \begin{align*}
	  m_1 & : w_1 \succ w_2 \succ w_3 & \qquad & w_1  : m_2 \succ m_3 \succ m_1\\
	  m_2 & : w_2 \succ w_3 \succ w_1 & \qquad & w_2  : m_3 \succ m_1 \succ m_2\\
	  m_3 & : w_3 \succ w_1 \succ w_2 & \qquad & w_3  : m_1 \succ m_2 \succ m_3
	 \end{align*}
	 Phase 1 of Irving's algorithm does not alter the preferences of this instance. 
	 Thus, the above preference profile is $P_0$.
	 In $P_0$, rotations $\varphi_1 = (m_1, w_1), (m_2, w_2), (m_3, w_3)$ and $\varphi_2 = (w_1, m_2), (w_2, m_3), (w_3,m_1)$ are exposed.
	 After eliminating $\varphi_1$, rotations~$\varphi_2$ and $\varphi_3 = (m_1, w_2), (m_2, w_3), (m_3, w_1)$ are exposed. 
	 After eliminating $\varphi_2$, rotations~$\varphi_1$ and $\varphi_4 = (w_1, m_3), (w_2, m_1), (w_3, m_2)$ are exposed. 
	 $\varphi_1,\varphi_2,\varphi_3,\varphi_4$ are the only rotations.
	 Note that $\bar \varphi_1 = \varphi_4$ and $\bar \varphi_2 = \varphi_3$, implying that all four rotations are nonsingular and that no singular rotation exists.
	 The rotation poset contains only the following two relations:
	 $\varphi_1 $ precedes $\varphi_3$ and $\varphi_2 $ precedes $\varphi_4$.
	 Consequently, there are three closed and complete subsets of the rotation poset:
	 $\{\varphi_1, \varphi_2\}$ (whose elimination results in the stable matching~$\{\{m_1, w_2\}, \{m_2, w_3\}, \{m_3, w_1\}\}$), $\{\varphi_1, \varphi_3\}$ (whose elimination results in~$\{\{m_1, w_3\}, \{m_2, w_1\}, \{m_3, w_2\}\}$), and $\{\varphi_2, \varphi_4\}$ (whose elimination results in~$\{\{m_1, w_1\}, \{m_2, w_2\}, \{m_3, w_3\}\}$).
	\end{example}

	We continue by observing the following basic fact about rotations:
\begin{lemma}[{\cite[p. 169 and Lemma 4.2.7]{DBLP:books/daglib/0066875}}]\label{obs}\label{obs2}
 If rotation $\varphi=(a_{i_1}, a_{j_1}), \dots, (a_{i_{r}}, a_{j_{r}})$ is exposed in some stable table~$T$, then $a_{i_k}$ is the last agent in the preferences of $a_{j_k}$ in~$T$ for each $k\in [r]$. Eliminating $\varphi$ in particular includes deleting the pair $\{a_{i_k}, a_{j_k}\}$ for each~$k\in [r]$. 
\end{lemma}

For our algorithm, we will exploit that it is possible to work on sets of rotations instead of stable matchings, as there is a bijection between closed complete subsets of rotations and stable matchings. 
In particular, give a closed and complete subset of rotations $Z$, there is an ordering of the rotations from~$Z$ such that starting with $P_0$ we can eliminate the exposed rotations one by one, resulting in a preference profile where the preferences of each agent $a$ only contain the partner of $a$ in the matching corresponding to $Z$ (see also \Cref{ex:rotation}):
\begin{lemma}[{\cite[Theorem 4.3.2]{DBLP:books/daglib/0066875}}]
	\label{lem:poset}
	There is a bijection between closed and complete subsets of rotations and stable matchings.
	The bijection maps each closed and complete subset~$Z$ of rotations to the matching arising through the elimination of each rotation of~$Z$.
\end{lemma} 

\subsubsection{High-Level Idea and Useful Lemmas}   \label{rotations:facts}

The general idea behind our algorithm for \ISRForcedForbidden is to successively alter the closed and complete set of rotations $Z_1$ corresponding to the given matching $M_1$ in order to include all forced and exclude all forbidden pairs. 
At the core of our algorithm lies the observation that rotations come with certain identifiable guarantees how ``good'' an agent is matched in a resulting stable matching:
For instance, in case we eliminate an exposed rotation that makes agent $c$ the last agent in the preferences of $a$ (recall \Cref{obs}), we know that $a$ is either matched to $c$ or an agent it prefers to $c$ in the corresponding stable matching.
This allows one to identify, for some agent pair $\{a,b\}$ certain (prohibited) rotations that if included in a set of rotations guarantee that the pair cannot be part of the corresponding stable matching (those rotations guarantee that $a$ is matched better than $b$). 
Conversely, there is often also a (necessary) rotation that needs to be included in a set of rotations corresponding to a stable matching containing the pair (the rotation that ensures that $a$ is matched better than all agents to which it prefers $b$).
These necessary and prohibited rotations then allow us to control whether pairs are (not) included in the output stable matching. 
For instance, in order to ensure that all forced pairs are contained in the matching, we alter $Z_1$ to include all necessary and exclude all prohibited rotations of forced pairs (thereby changing~$Z_1$ as little as possible to ensure that all forced pairs get included). 
For forbidden pairs, the situation will be slightly more complicated, as we can either not include the necessary rotation or include one of the prohibited rotations. 

In order to identify necessary and prohibited rotations, we start by stating a useful characterization under which circumstances and agent $b$ can become the last agent in the preferences of $a$ in some stable table due to
Gusfield~\cite{DBLP:journals/siamcomp/Gusfield88a}.
For this, for an agent pair $\{a,b\}$, let $\rho^{a,b}$ be the dual rotation of the rotation containing $(a,b)$ (if there is a stable table exposing a nonsingular rotation containing $(a,b)$). 
Considering \Cref{ex:rotation}, we have e.g.\ $\rho^{m_1, w_2} = \varphi_2$.
\begin{lemma}[{\cite[Corollary~5.1]{DBLP:journals/siamcomp/Gusfield88a}}] \label{lem:uniquerotation}
  Let $\{a, b \}$ be a stable pair such that there is a stable pair~$\{a, b'\}$ with $ a $ preferring~$b$ to~$b'$.
  Then, there is a rotation including $(a,b)$. 
  Moreover, $\rho^{a,b}$ is the unique rotation whose elimination makes $b$ the last choice of~$a$.
\end{lemma}

\Cref{lem:uniquerotation} directly implies that in case a closed and complete subset $Z$ contains $\rho^{a, b}$, agent $a$ cannot be matched worse than $b$ in the matching corresponding to $Z$:
\begin{lemma}
\label{lem:blocking-rotation}
  Let~$\{a, b\}$ be a stable pair such that there is a stable pair~$\{a, b'\}$ with $a$ preferring $b$ to~$b'$ and let $M$ be the stable matching corresponding to a closed and complete subset~$Z$.
  If $\rho^{a, b} \in Z$, then $\{a,b\}\in M$ or $a$ prefers $M(a)$ to $b$.
\end{lemma}
\begin{proof}
 If we eliminate $\rho^{a, b}$, then by \Cref{lem:uniquerotation}, agent~$b$ will become last in the preferences of $a$. 
 Thus, $a$ needs to be matched to~$b$ or better in the resulting matching.
\end{proof}

Combining \Cref{lem:uniquerotation,lem:poset} gives a characterization of when a pair~$\{a, b\}$ is contained in a stable matching:\footnote{\Cref{lem:cond-stable-pair} has been already implicitly used in the literature, e.g., \cite[Section~5]{DBLP:journals/tcs/FleinerIM07}, but we are not aware of an explicit formulation or proof of it.}
\begin{lemma}
\label{lem:cond-stable-pair}
  Let~$\{a, b\}$ be a stable pair such that there is a stable pair~$\{a, b'\}$ with $a$ preferring $b$ to~$b'$ and let $M$ be the stable matching corresponding to a closed and complete subset~$Z$.
  Then $\{a, b\} \in M$ if and only if $\rho^{a, b} \in Z$ and for any stable partner~$b^*$ which $a$ prefers to~$b$, we have $\rho^{a, b^*} \notin Z$. 
\end{lemma}
\begin{proof} 
 We start by proving the forward direction. 
 Let $M$ be a stable matching with $\{a,b\} \in M$ corresponding to the closed and complete subset $Z$ of rotations. 
 Successively eliminating rotations from $Z$ to arrive at matching $M$, at some point $b$ needs to become the last choice of $a$. 
 By \Cref{lem:uniquerotation} for this we need to eliminate rotation $\rho^{a,b}$, implying that $\rho^{a,b}\in Z$. 
 Moreover, note that in case we eliminate a rotation $\rho^{a, b^*}$ where $b^*$ is a stable partner of $a$ which $a$ prefers to $b$, then by \Cref{lem:uniquerotation} agent~$b^*$ becomes the last agent in the preferences of $a$. 
 As $a$ prefers $b^*$ to $b$, this implies that $b$ got deleted from the preferences of $a$, a contradiction.
 
 For the backwards direction, assume that $\rho^{a, b}\in Z$ and $\rho^{a, b^*} \notin Z$ for every stable partner~$b^*$ of~$a$ which $a$ prefers to~$b^*$.
 As $\rho^{a,b}\in Z$, \Cref{lem:blocking-rotation} implies that $a$ is matched at least as good as~$b$ in $M$. 
 Assume for the sake of contradiction that $a$ is matched to an agent $b^*$ it prefers to $b$ in $M$.
 However, for $b^*$ to become the only agent in the preferences of $a$ it in particular needs to become the last agent.
 By \Cref{lem:uniquerotation} this requires $\rho^{a,b^*}\in Z$, a contradiction.
\end{proof}
Going back to our initially described intuition, for stable pairs~$\{a,b\}$ covered by \Cref{lem:cond-stable-pair}, $\rho^{a,b}$ can be interpreted as the necessary rotation  and the rotations $\rho^{a,b^*}$ for all stable partners $b^*$ which $a$ prefers to $b$ can be interpreted as the prohibited rotations.
To give an example for this, consider again \Cref{ex:rotation}, and let us focus on the stable pair $\{m_1,w_2\}$. 
Agent $m_1$ has a stable partner~$w_1$ it prefers to~$w_2$ and a stable partner $w_3$ to which it prefers $w_2$. 
Thus, by \Cref{lem:cond-stable-pair}, for $\{m_1,w_2\}$ to be included in a stable matching, the corresponding set of rotations needs to include $\rho^{m_1,w_2}=\varphi_2$ and cannot include $\rho^{m_1,w_1}=\varphi_4$ (in fact the single stable matching containing $\{m_1,w_2\}$ corresponds to the rotation set $\{\varphi_1,\varphi_2\}$).

Finally, we conclude by observing that for every stable pair $\{a,b\}$
and each stable matching $N$ not including $\{a,b\}$ exactly one of $a$ and $b$ prefers
the other to its partner in $N$:
\begin{restatable}[{\cite[Lemma 4.3.9]{DBLP:books/daglib/0066875}}]{lemma}{circularprefs}\label{lem:circular-prefs}
	Let~$N$ be a stable matching and $e = \{a, b\} \notin N$ be a stable pair.
	Then either $N(a)\succ_{a} b$ and $a\succ_{b} N(b)$ or 
	$b\succ_{a} N(a)$ and $N(b)\succ_{b} a$.
\end{restatable}

\subsubsection{The Algorithm}  \label{alg}
Using the machinery from \Cref{rotations:facts}, we are now ready to present our algorithm. 

\begin{theorem}\label{thm:fpt-new}
	\ISRForcedForbidden can be solved in $\mathcal{O} (2^{|P\cap M_1|} \cdot n \cdot m) $ time.
\end{theorem}

\begin{proof}
    In the algorithm, we will guess\footnote{``Guessing'' can be interpreted as iterating over all possibilities.} for each forbidden pair~$e = \{a, b\} \in P \cap M_1$ whether $a $ or $b$ prefers its partner in the output matching to its partner in $M_1$.
    We say that a matching $M$ \emph{respects our guesses} if for each forbidden pair $e=\{a,b\}\in P\cap M_1$, $a$ prefers its partner in $M$ to $b$ if and only if we guessed that this is the case.
	We assume that there is at least one stable matching containing all forced and none of the forbidden pairs that respects our guesses, as we can reject the current guess otherwise (and this can be checked in $\mathcal{O}(m)$ time by reducing it to an instance of \textsc{Stable Roommates with Forced and Forbidden Pairs} \cite{DBLP:journals/tcs/FleinerIM07}).
	We further assume without loss of generality that $P$ only contains stable pairs (otherwise, we can delete the pair from $P$, as each stable matching will trivially not contain this pair). 
	
	In the following, when we say that we \emph{integrate a (nonsingular) rotation~$\varphi$} in a closed and complete set~$Z$ of rotations, then  we add~$\varphi$ and all rotations preceding $\varphi$ to $Z$ and delete~$\bar \varphi$ and all rotations preceded by~$\bar \varphi$ from~$Z$. 
	Before we present the algorithm, we now argue that after integrating a nonsingular rotation $\varphi$ to a closed and complete set~$Z$, the resulting set $Z'$ is still closed and complete: 
	$Z'$ is closed, as $Z$ is closed and in case we add a rotation we also add all its predecessors  and in case we delete a rotation we also delete all its successors.
	Moreover, $Z'$ is complete: 
	When integrating $\varphi$, we first add $\varphi$ and delete $\bar{\varphi}$.
	For all other rotations that we add, i.e., all rotations preceding $\varphi$, we delete their dual and for all ``dual'' rotations we delete, i.e., all rotations succeeding $\bar{\varphi}$, we add the ``primal'', as $\varphi \vartriangleright \rho$ if and only if $\bar \rho \vartriangleright \bar \varphi$ \cite[Lemma~4.3.7]{DBLP:books/daglib/0066875}.
	
	In the algorithm, we start with a closed and complete subset of rotations $Z$ and then only modify $Z$ by integrating rotations. Thus, $Z$ remains to be closed and complete over the course of the algorithm. We denote as $M_Z$ the stable matching corresponding to~$Z$ (the correspondence between matchings and sets of rotations is described in \Cref{lem:poset}). 
	
	\paragraph{The Algorithm.} Our algorithm works as follows:
	\begin{enumerate}
		\item Compute the rotation digraph which contains a vertex for each rotation and an arc from rotation $\varphi$ to rotation $\rho$ if $\varphi$ precedes $\rho$.
		Let $Z_1$ be the closed complete subset of rotations corresponding to~$M_1$, which exists and is unique by \Cref{lem:poset}.
		Set $Z:= Z_1$.
		
		\item For each forced pair $\{a,b\}\in Q$ that is not a fixed pair, assume without loss of generality that there is a stable pair~$\{a, b'\}$ with~$a $ preferring~$b$ to~$b'$ (for one of the two agents such a pair needs to exist by \Cref{lem:circular-prefs} and as $\{a,b\}$ is stable but not fixed). 
		We integrate rotation~$\rho^{a, b}$ to~$Z$.
		Further, for each stable pair~$\{a, b^*\}$ with $a$ preferring~$b^*$ to~$b$, we integrate~$\bar \rho^{a, b^*}$ to~$Z$.
		\label{item:forced}
		
		\item \label{item:guess}
		For each forbidden pair~$e = \{a, b\} \in P \cap M_1$, we guess whether $a $ or $b$ prefers its partner in the desired matching to its partner in $M_1$ (note that by \Cref{lem:circular-prefs}, exactly one of $a$ and $b$ has to do this).
		We assume without loss of generality that we guessed that $a$ prefers its partner in the desired matching to~$b$.
		Let $b^*$ be the least-preferred (by $a$) stable partner of~$a$ which $a$ prefers to~$b$ (such a partner needs to exist by our guess).
		Integrate~$\rho^{a, b^*}$ to~$Z$.
		
		\item \label{item:forbidden}
		As long as the matching~$M_Z$ contains a pair~$e = \{a, b\} \in P \setminus M_1$, assume without loss of generality that $a$ prefers~$M_Z (a)=b$ to~$M_1 (a)$ (for one of the two agents this needs to hold by \Cref{lem:circular-prefs}, as $\{a,b\}$ is a stable pair).
		Let $b^*$ be the least-preferred (by $a$) stable partner of~$a$ which $a$ prefers to~$b$. If $b^*$ exists, we integrate~$\rho^{a, b^*}$ to~$Z$; otherwise we do nothing.
		
		\item Return the matching~$M:=M_Z$.
	\end{enumerate}
	
	\paragraph{Proof of Correctness.} We start by showing that all changes made to $Z$ over the course of the algorithm are indeed necessary.
	\begin{claim}\label{claim:FPT-algo}
		Let~$M^*$ be a stable matching containing all forced pairs and no forbidden pairs which respects our guesses.
		Further, let $Z^*$ be the corresponding closed and complete subset of rotations.
		Then $Z^*$ contains all rotations added in \Cref{item:forced,item:guess,item:forbidden}.
		Moreover, the agent $b^*$ defined in \Cref{item:forbidden} always exists.
	\end{claim}
	
	\begin{claimproof}
        We start by proving the first part of the claim. 
        Note that it is sufficient to prove the statement for all integrated rotations, as all other rotations $\rho$ that we added to~$Z$ precede an integrated rotation.
        
		Each rotation integrated in \Cref{item:forced} is contained in $Z^*$ by \Cref{lem:cond-stable-pair} and as $Z^*$ needs to be complete.
		
		Next, we consider the rotations integrated in \Cref{item:guess}.
		Assume without loss of generality that we guessed that $a$ prefers~its partner in the desired matching to $b$.
		Then $a$ must be matched at least as good as~$b^*$ in the desired matching.
		Assume towards a contradiction that $\rho^{a, b^*} \notin Z^*$, implying $\bar \rho^{a, b^*} \in Z^*$.
		Rotation~$\bar \rho^{a, b^*}$ contains~$(a, b^*)$ (by definition of $\bar \rho^{a, b^*}$ and as the dual of $\bar \rho$ is again $\rho$). 
		Recall that in case $\bar \rho^{a, b^*}$ is exposed, then $a$ is the last choice of $b^*$ (\Cref{obs}) and eliminating the rotation implies deleting the pair $\{a,b^*\}$ (\Cref{obs2}). 
		Thus, after the elimination of $\bar \rho^{a, b^*}$, agent~$b^*$ prefers its last choice to~$a$.
		This implies that $b^*$ prefers~$M^* (b^*)$ to~$a$.
		By \Cref{lem:circular-prefs} and as $\{a,b^*\}$ is a stable pair not contained in $M^*$, it follows that $a $ prefers~$b^*$ to~$M ^* (a)$.
		This is a contradiction to~$ a $ preferring~$M^* (a)$ to~$b$ and the definition of~$b^*$.
		Consequently, $\rho^{a, b^*}  \in Z^*$.
		
		For \Cref{item:forbidden}, we show by induction that the claim holds after the $i$-th execution of this step.
		The statement clearly holds before the first execution of the step.
		Let~$Z^i$ be the set $Z$ before the $i$-th execution.
		Let~$\{a, b\}$ be the pair examined in this execution, with $a$ preferring~$b$ to~$M_1 (a)$.
		\Cref{lem:cond-stable-pair} implies that $\rho^{a, b} \in Z^i$, as $a$ is matched to $b$ in $M_{Z^i}$ and $M_1(a)$ is a stable partner of $a$ to which $a$ prefers $b$.
		Moreover, we need to have that $\rho^{a, b} \notin Z_1$: If $\rho^{a, b} \in Z_1$, then by \Cref{lem:blocking-rotation}, $a$ is matched at least as good as $b$ in $M_1$, contradicting our assumption that $a$ prefers $b$ to $M_1(a)$.
		By our induction hypothesis it follows that $\rho^{a,b}\in Z^*$. 
		Applying again \Cref{lem:blocking-rotation} it follows that $a$ is matched at least as good as $b$ in $M^*$. 
		As $\{a,b\}$ is a forbidden pair, we even get that $a$ prefers $M^*(a)$ to $b$. 
		The remainder of the proof is now analogous to \Cref{item:guess}.

    Concerning the second part of the claim, observe that we have established above that in each iteration of \Cref{item:forbidden}, $a$ prefers $M^*(a)$ to $b$. From this it follows that $a$ has a stable partner it prefers to $b$ and thus in particular that $b^*$ exists in each execution of \Cref{item:forbidden}.
	\end{claimproof}
	
	Recall that we have assumed that there is a stable matching~$M^*$ containing all forced and no forbidden pairs that respects our guesses.
	Thus, as all rotations integrated to~$Z$ must be contained in~$Z^*$ by \Cref{claim:FPT-algo}, there is no rotation~$\rho$ such that $\rho$ as well as $\bar \rho$ get added to~$Z$ during the algorithm (as in this case, $Z^*$ would not be a complete subset of rotations).
	\Cref{item:forced} now ensures by \Cref{lem:cond-stable-pair} that $M$ contains all forced pairs.
	\Cref{item:guess,item:forbidden} ensure that $M$ contains no forbidden pair by \Cref{lem:blocking-rotation} (note that the case that $b^*$ does not exist in \Cref{item:forbidden} never occurs as proven in \Cref{claim:FPT-algo}).
	
	Next, we show the optimality of~$M$.
	Let~$Z^*$ be the subset of the rotation poset corresponding to~$M^*$.
	By \Cref{claim:FPT-algo}, we get 
	$Z_1 \triangle Z \subseteq Z_1 \triangle Z^*$ (\Cref{claim:FPT-algo} directly implies that $Z\setminus Z_1\subseteq Z^*\setminus Z_1$ but also gives us $Z_1\setminus Z\subseteq Z_1\setminus Z^*$ as deleting a rotation corresponds to adding its dual).
	We now show that we can conclude from this that there is no pair~$e \in (M_1 \cap M^*) \setminus (M_1 \cap M)$:
	Assume towards a contradiction that there is some~$e = \{a, b\} \in (M_1 \cap M^*) \setminus (M_1 \cap M)$.
	Note that as $\{a,b\}$ is not contained in the stable matching $M$, it is not a fixed pair.
	Assume without loss of generality that there is a stable pair~$\{a, b'\}$ with $a$ preferring~$b$ to~$b'$ (for one of the two agents this needs to exist by \Cref{lem:circular-prefs}, as $\{a,b\}$ is a stable pair not contained in $M$).
	Thus, by \Cref{lem:cond-stable-pair}, $Z_1 \cap Z^*$ contain~$\rho^{a, b}$ as well as $\bar \rho^{a, b^*}$ for any stable partner~$b^*$ which $a$ prefers to~$b$.
	Since $Z_1 \cap Z^* \subseteq Z_1 \cap Z$ these rotations are also contained in $Z$  and by \Cref{lem:cond-stable-pair} it follows  that $\{a, b\}$ is also contained in $M$, a contradiction to $\{a, b\} \in (M_1 \cap M^*) \setminus (M_1 \cap M)$.
	
	\paragraph{Running Time.}
	Computing the rotation digraph can be done in $\mathcal{O}( n \cdot m)$ time~\cite{DBLP:journals/algorithmica/Feder94}.
	In \Cref{item:guess}, there are $2^{|P \cap M_1|}$ guesses.
	For each guess, any pair can be added at most once to~$M_Z$ and any rotation can be added at most once to~$Z$.
	Thus, the remaining part of \Cref{item:forced,item:guess,item:forbidden} can be done in $\mathcal{O}(m)$ total time.
	Consequently, the algorithm runs in~$\mathcal{O}\bigl((2^{|P \cap M_1|} + n) \cdot m\bigr)$~time.
\end{proof}

\subsection{(FPT-)Algorithm for Adapt Strongly SR with Ties to Forced and Forbidden Pairs}\label{fpt:strong}

In case of strong stability, we can employ a similar algorithm as for strict preferences, as this problem also admits a (slightly more complicated) rotation poset.
Although the definition of the rotations and their duals differ from the ``classical'' case without ties, they still fulfill crucial properties exploited in \Cref{thm:fpt-new}:
\begin{enumerate} 
	\item Analogous to \Cref{lem:poset}, each stable matching corresponds to a closed and complete subset of the rotation poset.
	\item \label{item:edge}
	Somewhat analogous to \Cref{lem:cond-stable-pair},
	for each stable pair~$e$ there are two rotations~$\rho$ and $\varphi$ such that $e$ may be contained in a stable matching corresponding to a complete and closed set~$Z$ of rotations if and only if~$\rho \in Z$ and $\varphi \notin Z$.
\end{enumerate}
However, for strong stability, there are now multiple possible stable matchings corresponding to the same set of rotations (as rotations here only encode the rank of the partner of an agent in the matching instead of the partner itself).
In order to solve \ISRStrongtiesForcedForbidden, it in fact suffices to compute the closed and complete set~$Z$ of rotations corresponding to an optimal solution (as subsequently we can find the stable matching corresponding to $Z$ closest to $M_1$ using a minimum-cost matching algorithm).
Turning to the constraints that forced and forbidden pairs impose on $Z$, as for the case of strict preferences, each forced pair gives rise to the constraint that one rotation is contained and one rotation is not contained in~$Z$ (due to part (2) of the above enumeration).
For forbidden pairs, however, the situation is different and more complicated:
As there may be multiple stable matchings for the same set of rotations with only some of them not containing a forbidden pair in question, a forbidden pair does not necessarily lead to a constraint on~$Z$ (even if this forbidden pair is also contained in~$M_1$).
In order to be able to solve the problem, we show as a crucial step that we can determine whether (a set of) forbidden pairs lead to constraints on~$Z$.
Overall, we can show the following: 
\begin{restatable}[\appsymb]{theorem}{SMstrong}\label{thm:fpt-strong}
	\textsc{Adapt Strongly SR with Ties to Forced and Forbidden Edges} can be solved in $\mathcal{O} \bigl( (2^{|P\cap M_1|} + m) \cdot \sqrt{n} m \log n\bigr)$ time.
\end{restatable}

\section{Conclusion}
We have conducted a complete and fine-grained
analysis of minimally changing stable matchings to include forced and exclude forbidden pairs. 
As our main result, we have proven that in the absence of ties in the agent's preferences \ISRForcedForbidden
is fixed-parameter tractable with respect to the number of forbidden pairs in the given matching (and thus polynomial-time solvable if there are only forced pairs). 
At the core of this algorithm lies a clever exploitation of the rotation poset that might inspire similar approaches for related problems.
For example, one might want to adapt a matching to (dis)satisfy certain groups of agents or to improve the situation of the agent which is worst of. 
All these requirements can be encoded if we are given for each agent an upper and lower bound for how the agent is matched in the new matching $M_2$ (i.e., for each agent~$a\in A$, we are given two agents $b_a$ and $v_a$ and we require $b_a\succ_a M_2(a) \succ_a v_a$).
Slightly adapting the initialization procedure of our algorithm (by starting with incorporating rotations realizing these constraints) this problem becomes polynomial-time solvable.

To the best of our knowledge, the idea of minimally changing a given matching to incorporate external requirements has not been studied in previous works.
Thus, extending our studies of forced and forbidden pairs for stable matchings to other incremental requirements such as group fairness or diversity constraints or other matching problems such as popular matching is an interesting direction for future work.

\newpage

\appendix

\newpage

\section{Additional Material for Subsection \ref{fpt:strong}}
This section is devoted to proving the following result: 
\SMstrong*

In the following we assume that all considered stable matchings (and in particular the initial matching $M_1$) are complete matchings, as we can otherwise modify the instance accordingly in $\mathcal{O}(m)$ time.\footnote{As in the case with strict preferences, if $M_1$ is not complete, let $B$ be the set of agents unmatched in $M_1$. For each agent $b\in B$, we add an agent $b'$ to the instance which only finds $b$ acceptable and which is added at the end of the preferences of $b$ (so $b$ prefers all other agents it accepts strictly to $b
$). Then, using the Rural Hospitals Theorem for Strongly SR \cite{Manlove99}, which says that each stable matching in a \textsc{Strongly SR with Ties} instance matches the same set of agents, it follows that all stable matchings in the modified instance contain pairs $\{\{b,b'\}\mid b\in B\}$. From this the correctness follows.} 

For agents $a,a'\in A$ with $a'\in \Ac(a)$, $\rk_a(a')$ denotes the rank of $a'$ in the preferences of $a$, i.e., the number of agents $a$ strictly prefers to $a'$ plus one. 

The rest of this section is structured as follows. In \Cref{rotation-StrongSM}, we explain what rotations are in \textsc{Strongly SR with Ties} instances and how they relate to rotations in an ``equivalent'' \textsc{Strongly SM with Ties} instance. 
Subsequently, in \Cref{sec:necprostrong}, we establish that pairs have a necessary rotation (that needs to be included in a rotation set for the pair to be part of some stable matching corresponding to the set) and a prohibited rotation (that cannot be included in a rotation set for the pair to be part of some stable matching corresponding to the set). 
Afterwards, in \Cref{strong:forbidden}, we establish some useful lemmas that will help us to deal with identifying the constraints imposed by forbidden pairs.
Finally, in \Cref{strong:alg} we present our algorithm (that shares some similarities with the algorithm from \Cref{thm:fpt-new}) and prove its correctness. 

\subsection{Relation between \textsc{Strongly SM with Ties} and \textsc{Strongly SR with Ties}: Definitions and Rotations~\cite{DBLP:conf/esa/Kunysz16}} \label{rotation-StrongSM}
To understand the structure of stable matchings in \textsc{Strongly SR with Ties} instances,  it is helpful to understand \textsc{Strongly SM with Ties}.
The reason for this is that one may ``reduce'' \textsc{Strongly SR with Ties} to \textsc{Strongly SM with Ties} by replacing each agent~$a$ by a man~$a^m$ and a woman~$a^w$ (with the same preferences, i.e., if the preferences of~$a$ are $a_1 \succ a_2 \succ \dots \succ a_k$, then the preferences of~$a^m$ are $a_1^w \succ a_2^w \succ \dots \succ a_k^w$ and the preferences of~$a^w$ are $a_1^m \succ a_2^m \succ \dots \succ a_k^m$).
For an instance~$\mathcal{I}$ of \textsc{Strongly SR with Ties}, we will call the resulting instance of \textsc{Strongly SM with Ties} $\mathcal{I}^M$.
Then, stable matchings in~$\mathcal{I}$ correspond to ``symmetric'' stable matchings in~$\mathcal{I}^M$ (these are matchings where for each $a\in A$ there is an $b\in A$ such that both~$a^w$ is matched to $b^m$ and $a^m$ is matched to $b^w$).
Thus, we start by reviewing important definitions for \textsc{Strongly SM with Ties}.

	Consider an instance $(U,W,\succsim_{a\in U\cup W})$ of \textsc{Strongly SM with Ties}.
	Let $\ge^{\rk}$ be the partial order on the set of (stable) matchings with $M\ge^{\rk} N$ if and only if $\rk_m (M(m)) \ge^{\rk} \rk_m (N(m))$ for every man~$m\in U$.
	This induces an equivalence relation~$\sim^{\rk} $ on the (stable) matchings with $M \sim^{\rk} N$ if and only if $M\ge^{\rk} N$ and $N\ge^{\rk} M$.
	In the following, in case we speak of an equivalence class of matchings, we always refer to the equivalence classes induced by $\sim^{\rk}$.
	A \emph{strict successor} of a stable matching~$M$ is a stable matching~$N$ with $N >^{\rk} M$ such that there is no matching~$M'$ with $N >^{\rk} M' >^{\rk} M$.
	A \emph{rotation transforming~$M$ into a strict successor~$N$} is the set~$\rho := \{\big(a, \rk_a (M(a)), \rk_a (N(a)) \big) \mid \forall a\in U\cup W: M(a) \neq N(a)\}$. Note that we have by definition that $\rk_m (M(m))<\rk_m (N(m))$ for each man $m\in U$.
	In the following, we call a set of triples a rotation (in an instance) if there exists two matchings $M$ and $N$ with $N$ being a strict successor of $M$ and the rotation transforms $M$ into $N$.
	
	A \emph{maximal} sequence of stable matchings is a sequence~$M_1, \dots, M_k$ of stable matchings where $M_1$ is a men-optimal matching,\footnote{We call a matching men-optimal if each man is matched to a best stable partner and women-optimal if each women is matched to a best stable partner.} $M_k$ is a women-optimal matching, and for~$i \in [k-1]$, matching~$M_{i+1}$ is a strict succesor of~$M_i$.
An important property of rotations is that for any maximal sequence~$M_1, \dots, M_k$ of stable matchings all rotations of the instance appear as a rotation transforming $M_i$ to~$M_{i+1}$ for some~$i \in [k-1]$~\cite{DBLP:conf/soda/KunyszPG16}.
This allows one to define a partial order on the set of rotations by saying that a rotation~$\rho $ \emph{precedes} a rotation~$\pi$ if for every maximal sequence~$M_1, \dots, M_k$ of stable matchings, we have that $\rho$ transforms~$M_i$ to~$M_{i+1}$ and $\pi$ transforms~$M_j$ to $M_{j+1}$ with~$i < j$.
Note that for two rotation~$\rho$ and $\pi$ with $(m, x_\rho ,y_\rho) \in \rho$ and $(m, x_\pi , y_\pi) \in \pi$ with $x_\rho < x_\pi$ for some man~$m$, rotation~$\rho $ precedes~$\pi$.
For a man $m$, let $\rk^*_m$ be the rank that $m$ has for his partner in the men-optimal matchings.
We say that a set of rotations is \emph{closed} if for each included rotation all predecessors are also included. 

A crucial theorem is the following, establishing the correspondence between stable matchings in an instance of \textsc{Strongly SM} and the rotation poset:
\begin{theorem}[\cite{DBLP:conf/soda/KunyszPG16}]
\label{thm:strong-rotation-equivalence}
	There is a one-to-one-correspondence between equivalence classes of stable matchings in a \textsc{Strongly SM with Ties} instance and closed subsets of the rotation poset.
	This bijection maps each closed rotation subset~$Z$ to the class of matchings~$M$ with $\rk_m (M(m)) = \max (\{ y\mid \exists x: (m, x, y) \in Z\} \cup\{\rk^*_m\})$ for each man~$m\in M$.
\end{theorem}

Note that the partial order $\ge^{\rk}$ formerly defined for \textsc{Strongly SM with Ties} naturally extends to \textsc{Strongly SR with Ties} by saying that $M\ge^{\rk} N$ if and only if $\rk_a (M(a)) \ge^{\rk} \rk_a (N(a))$ for each agent~$a\in A$.
\Cref{thm:strong-rotation-equivalence} can be extended to a correspondence between strongly stable matchings in the \textsc{Strongly SR with Ties} instance~$\mathcal{I}$ and the rotations of~$\mathcal{I}^M$.
For this, we first need to define the dual of a rotation $\rho$:
For a rotation $\rho $ in $\mathcal{I}^M$, its \emph{dual}~$\bar \rho$ is $\{(a^w, s, f) \mid  (a^m, f, s) \in \rho\} \cup \{ (a^m, s, f) \mid (a^w, f, s) \in \rho\}$.
Kunysz~\cite[Theorem~26]{DBLP:conf/esa/Kunysz16} showed if $\mathcal{I}$ admits a strongly stable matching, then the dual of each rotation in $\mathcal{I}^M$ is again a rotation in $\mathcal{I}^M$.
A subset of rotations of $\mathcal{I}^M$ is \emph{complete} if it contains for each rotation~$\rho$ from $\mathcal{I}^M$ either $\rho$ or $\bar \rho$.
The crucial observation here is now that each closed and complete subsets of rotations in $\mathcal{I}^M$ corresponds to an equivalence class containing a symmetric stable matching in $\mathcal{I}^M$ and thus also to an equivalence class of stable matchings in $\mathcal{I}$.
This allows one to conclude the following:

\begin{theorem}[\cite{DBLP:conf/esa/Kunysz16}]\label{poset:equiv}
	There is a one-to-one correspondence between equivalence classes  of stable matchings in a \textsc{Strongly SR with Ties} instance~$\mathcal{I}$ and closed and complete subsets of the rotation poset for~$\mathcal{I}^M$. 
	This bijection maps each closed and complete rotation subset~$Z$ in $\mathcal{I}^M$ to the class of matchings~$M$ in $\mathcal{I}$ with $\rk_a (M(a)) = \max (\{ y\mid \exists x: (a^m, x, y) \in Z\} \cup\{\rk^*_{a^m}\})$ for each agent~$a\in A$.
\end{theorem}
Accordingly, each closed and complete subset $Z$ of rotations (in $\mathcal{I}^M$) corresponds to a set of matchings $\mathcal{M}_Z$ in $\mathcal{I}$. 
For each agent $a\in A$ it holds that it has the same rank for its partner in all matchings from $\mathcal{M}_Z$. We refer to this as  $\rk_a(Z)$.
Let $G^Z$ be the graph containing the agent set as vertices with an edge between two agents $\{a,a'\}$ if and only if $\rk_a(a')=\rk_a(Z)$ and $\rk_{a'}(a)=\rk_{a'}(Z)$. 
Then, we have that the matchings from $\mathcal{M}_Z$ are exactly the complete matchings in $G^Z$.
Instead of $G^Z$ we sometimes also write $G^{\mathcal{M}_Z}$.
\newcommand{\prh}{\text{prh}}
\newcommand{\nec}{\text{nec}}

\subsection{Necessary and Prohibited Rotations}\label{sec:necprostrong}
In this subsection, we show that somewhat analogous to \Cref{lem:cond-stable-pair},
	for each stable pair~$e$ there are two rotations~$\varphi^e_{\nec}$ and $\varphi^e_{\prh}$ such that $e$ may be contained in a stable matching corresponding to a complete and closed set~$Z$ of rotations if and only if~$\varphi^e_{\nec} \in Z$ and $\varphi^e_{\prh} \notin Z$. This makes it easy to include forced pairs in a matching.
The following observation is useful for proving this: 
\begin{observation}[{\cite[Lemma~2]{DBLP:journals/dam/Manlove02}}]
\label{obs:indifference}
  Let~$M$ and $M'$ be two stable matchings in an instance of \textsc{Strongly SM with Ties}.
  If agent~$a$ is indifferent between $M(a) $ and $M' (a)$,
  then $M(a)$ is indifferent between~$a$ and~$M' (M (a))$.
  If agent~$a$ strictly prefers~$M(a) $ to~$M' (a)$,
  then $M(a)$ strictly prefers~$M' (M (a))$ to~$a$.
\end{observation}

This together with the observation that for each man, there is at most one rotation moving its last choice to a given woman~\cite{DBLP:conf/soda/KunyszPG16} implies the existence of necessary and prohibited rotations:

\begin{lemma}\label{le:nec-phr}
 Let~$\mathcal{I}$ be an instance of \textsc{Strongly SM with Ties} and $e = \{m, w\}$ with~$m\in U$ and $w\in W$ be a stable pair that is not fixed.
 Let $\varphi^e_{\nec}$ be the rotation containing~$(m, x_{\nec}, y_\nec)$ with $y_\nec=\rk_m (w)\neq x_\nec$ (or $\varphi^e_{\nec}:=\emptyset$ if $w$ is the most preferred stable partner of $m$), and let $\varphi^e_{\prh}$ be the rotation containing $(m, x_{\prh}, y_\prh)$ with $x_\prh = \rk_m (w) \neq y_{\prh}$ (set $\varphi^e_{\prh}:=\emptyset$ if $w$ is the least preferred stable partner of $m$).
 
 For each closed set~$Z$ of rotations, there exists a stable matching~$M$ corresponding to~$Z$ with~$e \in M$ if and only if $\varphi^e_{\nec} \in Z$ and $\varphi^e_{\prh} \notin Z$. 
\end{lemma}

\begin{proof}
 By the definition of rotations (and as the set of rotations is independent of the chosen maximal sequence of stable matchings~\cite[Section 3.2]{DBLP:conf/soda/KunyszPG16}), there is at most one rotation containing~$\varphi^e_{\nec} = (m, x_{\nec}, y_\nec)$ with $y_\nec=\rk_m (w)\neq x_\nec$ and at most one rotation~$\varphi^e_{\prh}$ containing $(m, x_{\prh}, y_\prh)$ with $x_\prh = \rk_m (w) \neq y_{\prh}$.
 If $w $ is the best stable partner of~$m$, then $\varphi^e_\nec$ does not exist.
 If $w$ is the worst stable partner of~$m$, then $\varphi^e_\prh$ does not exist.
 In the following, we assume that $w$ is neither the best nor the worst stable partner of~$m$; the other two cases are analogous.
 
 Let~$Z$ be a closed set of rotations with~$\varphi^e_{\nec} \in Z$ and $\varphi^e_{\prh}\notin Z$.
 \Cref{thm:strong-rotation-equivalence} tells us how~$m$ may be matched:
 Because~$\varphi^e_{\nec} \in Z$, man~$m$ cannot be matched better than~$w$.
 Because~$\varphi^e_{\prh} \notin Z$ (and as thereby $Z$ also cannot contain any successors of $\varphi^e_{\prh}$), man~$m$ cannot be matched worse than~$w$.
 Thus, $m$ is indifferent between~$w$ and its partner in any matching corresponding to~$Z$.
 Let~$M^e$ be a stable matching containing~$e$ and let~$M^Z$ be a stable matching corresponding to~$Z$.
 By the Rural Hospitals Theorem~\cite{Manlove99} each strongly stable matching matches the same set of agents.
 Thus, $M^e \oplus M^Z$ consists of cycles.
 Let~$C$ be the cycle in $M^e \oplus M^Z$ containing~$e$.
 Since $m$ is indifferent between its neighbors in~$C$, from \Cref{obs:indifference} it follows that all agents in the cycle are indifferent between their neighbors in~$C$ .
 Consequently, $M^Z \oplus C$ is a strongly stable matching containing~$e$ which corresponds to $Z$ (as for no agent the rank of his partner changes when moving from $M^Z$ to $M^Z \oplus C$).
 
 Let~$Z$ be a closed set of rotations with~$\varphi^e_{\nec} \notin Z$ or $\varphi^e_\prh \in Z$.
 If $\varphi^e_{\nec} \notin Z$, then any stable matching~$M^Z$ corresponding to~$Z$ must match~$m$ to an agent of rank at most~$x_{\nec}$ by \Cref{thm:strong-rotation-equivalence} and the fact that no rotation succeeding~$\varphi^e_\nec$ is contained in~$Z$.
 Since $x_\nec < \rk_m (w)$, it follows that $e \notin M^Z$.
 If $\varphi^e_{\prh} \in Z$, then any stable matching~$M^Z$ corresponding to~$Z$ must match~$m$ to an agent of rank at least~$y_{\prh}$ by \Cref{thm:strong-rotation-equivalence}.
 Since $y_\prh > \rk_m (w)$, it follows that $e \notin M^Z$.
\end{proof}

We get the following easy consequence.

 \begin{lemma}\label{lem:same-nec-prh}
  Let $F$ be a maximal set of edges having the same necessary rotation~$\varphi_\nec$ and prohibited rotation~$\varphi_\prh$ and let~$X$ be the set of all endpoints of~$F$.
  Then any stable matching corresponding to a closed and complete subset~$Z$ containing $\varphi_\nec$ and $\bar \varphi_\prh$ contains a perfect matching on~$X$.
  Further, for each edge~$e$ of this perfect matching, we have that $\varphi_\nec^e = \varphi_\nec$ and $\varphi^e_\prh = \varphi_\prh$.
 \end{lemma}

 \begin{proof}
  Assume towards a contradiction that there is some matching~$M$ corresponding to~$Z$ which does not induce a perfect matching on~$X$.
  Then $M$ contains an edge~$e= \{a, b\}$ with~$a \in X $ and $b \notin X$.
  We assume that $a \in U$ and $b  \in W$;
  the case that $ b \in U$ and $ a\in W$ is symmetric by swapping the roles of men and women.
  Due to the maximality of~$F$, there is some edge~$f \in F$ incident to~$a$.
  Because $a$ is indifferent between~$e$ and~$f$ (as from $\varphi_\nec,\bar \varphi_\prh\in Z$ we get that $a$ is indifferent between $b$ and its partner in $M$), it follows that $e$ has the same necessary and prohibited rotation as~$f$, a contradiction to the definition of~$X$.
 \end{proof}
 
We extend the notions of necessary and \forrot\ rotations to \textsc{Strongly SR with Ties} as follows:
The rotation poset of a \textsc{Strongly SR with Ties} instance~$\mathcal{I}$ is the rotation poset of the \textsc{Strongly SM with Ties} instance~$\mathcal{I}^M$.
Each pair~$\{a, b\}$ from~$\mathcal{I}$ corresponds to two pairs~$\{a^m, b^w\}$ and $\{b^m , a^w\}$ from~$\mathcal{I}^M$.
Consequently, there are two different natural ways to define the necessary and \forrot\ rotations for~$e$:
We may pick the necessary and \forrot\ rotation for~$\{a^m, b^w\}$ or for $\{b^m, a^w\}$.
Thus, we define~$\varphi_\nec^{e, a} := \varphi_\nec^{\{a^m, b^w\}}$, $\varphi_\prh^{e, a} := \varphi_\prh^{\{a^m, b^w\}}$, $\varphi_\nec^{e, b} := \varphi_\nec^{\{b^m, a^w\}}$, and $\varphi_\prh^{e, b} := \varphi_\prh^{\{b^m, a^w\}}$.
Note that a closed and complete set~$Z$ of rotations contains $\varphi_\nec^{e, a}$ and not $\varphi_\prh^{e,a}$ if and only if $Z$ contains $\varphi_\nec^{e,b}$ and not $\varphi^{e, b}_\prh$ since the necessary rotation for~$\{a^m, b^w\}$ is the dual rotation of the \forrot\ rotation for~$\{b^m, a^w\}$.
Consequently, we say that two pairs~$e=\{a, b\}$ and $e'= \{a', b'\}$ have the same necessary and \forrot\ if either~$\varphi_\nec^{e, a} = \varphi_\nec^{e', a'}$ and $\varphi_\prh^{e, a} = \varphi_\prh^{e', a'}$ or~$\varphi_\nec^{e, a} = \varphi_\nec^{e', b'}$ and $\varphi_\prh^{e, a} = \varphi_\prh^{e', b'}$.
The reason why we distinguish between the necessary and \forrot\ with respect to~$\{a^m, b^w\}$ or $\{b^m, a^w\}$ is that this allows us to control whether $a$ or~$b$ improves when integrating~$\bar \varphi_\nec^{e, \cdot}$.
\subsection{Dealing with Forbidden Pairs} \label{strong:forbidden}
Before we turn to describing our algorithm, we describe how to treat forbidden pairs. 
Recall that as there may be multiple stable matchings for the same set of rotations with only some of them containing a forbidden pair in question, a forbidden pair does not necessarily lead to a constraint on~$Z$, even if this forbidden pair is also contained in~$M_1$.
In order to be able to solve the problem, in the following, we show as a crucial step that we can determine whether (a set of) forbidden pairs lead to constraints on~$Z$.

Using \Cref{obs:indifference}, we start by giving some structure to each stable matching in some equivalence class:
\begin{lemma}
\label{lem:A-sim}
 Let $\mathcal{M}$, $\mathcal{M}'$ be two equivalence classes such that there is an agent~$a^* \in A$ with~$M(a^*) \sim_{a^*} M'(a^*)$ for some~$M \in \mathcal{M}$ and $M' \in \mathcal{M}'$.
 Let~$A^{\sim} := \{a \in A \mid M(a) \sim_a M' (a)\}$.
 Then,
 \begin{enumerate}
  \item every stable matching in~$\mathcal{M}$ and $\mathcal{M}'$ induces a complete matching on~$A^\sim$.
  \item if $F$ is a minimal set of such that every matching from $\mathcal{M}$ which contains all forced pairs contains one pair from $F$, then $ F \cap E(G^{\mathcal{M}}[A^{\sim}])$ or $ F\setminus E(G^{\mathcal{M}}[A^{\sim}])$ is also such a minimal set.
 \end{enumerate}
\end{lemma}

\begin{proof}
Note that $M(a^*) \sim_{a^*} M'(a^*)$ for some~$M \in \mathcal{M}$ and $M' \in \mathcal{M}'$ in fact implies that $M(a^*) \sim_{a^*} M'(a^*)$ for all~$M \in \mathcal{M}$ and $M' \in \mathcal{M}'$.

 We start by proving the first part.
 Let~$a \in A^\sim$ and assume towards a contradiction that there is some~$M \in \mathcal{M}$ with $\{a, b\} \in M$ for some~$b\notin A^\sim$.
 For every~$M' \in \mathcal{M}'$, \Cref{obs:indifference} implies that $b$ is indifferent between~$M(b)$ and $M' (b)$, implying that $b \in A^\sim$, a contradiction to $b \notin A^\sim$.
 Thus, agents from~$A^\sim$ can only be matched to other agents from~$A^\sim$.
 Since each stable matching is complete due to our preprocessing step, the first part of the lemma follows.
 
 Using the first part, we now prove the second part.
 From the first part, it in particular follows that for two matchings $N,N'\in \mathcal{M}$ containing all forced pairs we have $(N\cap  E(G^{\mathcal{M}}[A^{\sim}]))\cup (N' \cap E(G^{\mathcal{M}}[A\setminus A^{\sim}]))$ is also a matching containing all forced pairs in $\mathcal{M}$. 
	Consequently, 	
	$ F \cap E(G^{\mathcal{M}}[A^{\sim}])$ or $ F\setminus E(G^{\mathcal{M}}[A^{\sim}])$ is a minimal set of pairs such that every stable matching in~$\mathcal{M}$ contains a pair from the set
	(if $F$ contains pairs from both $E(G^{\mathcal{M}}[A^{\sim}])$ and $E(G^{\mathcal{M}}[A\setminus A^{\sim}])$, pairs from one of the two sets are not necessary to cover all matchings).
\end{proof}

In order to check whether certain forbidden pairs impose constraints on the selected subset of rotations,  we prove that, for each minimal set~$F$ of pairs such that every stable matching which contains all forced pairs and belongs to some equivalence class contains at least one pair of~$F$, it holds that the necessary and \forrot\ are the same for each pair of~$F$.

We will call a stable matching which contains all forced pairs a \Pvalid{} matching.

\begin{lemma}
\label{cor:nec-prh-rot-F}
 Let~$F$ be a minimal set of pairs such that, for every \Pvalid{} matching $M$ which belongs to some equivalence class~$\mathcal{M}$, matching~$M$ contains at least one pair of~$F$. 
 Then, the necessary and \forrot\ rotation are the same for each pair of~$F$.
\end{lemma}

\begin{proof}
Note that in case $\mathcal{M}$ does not contain a matching including all forced pairs, $F$ is the empty set and thus the lemma statement is trivially fulfilled.
Thus, assume in the following that $\mathcal{M}$ contains at least one matching including all forced pairs.
 We start the proof by proving the following claim:
\begin{claim}
\label{lem:F}
 For every equivalence class~$\mathcal{M}'$ of stable matchings such that there is a \Pvalid{} matching~$M'\in \mathcal{M}' $ and~$e \in M'$ for some~$e \in F$, it holds that for every pair~$e' \in F$ there is some \Pvalid{} matching~$M''\in \mathcal{M}'$ with~$e' \in M''$.
\end{claim}

\begin{claimproof}
  Let $\mathcal{M}'$ be an equivalence class of stable matchings such that there is a \Pvalid{} matching~$M'\in \mathcal{M}' $ with~$e \in M'$ for some~$e \in F$.
  Assume towards a contradiction that there is some pair~$e' \in F$ with $e' \notin M''$ for every \Pvalid{} matching~$M'' \in \mathcal{M}'$.
  
  Let~$M \in \mathcal{M}$ be a matching containing all forced pairs.
  Let~$A^\sim := \{a \in A \mid M(a) \sim_a M' (a)\}$.
  \Cref{lem:A-sim} implies that $ F \cap E(G^{\mathcal{M}}[A^{\sim}])$ or $ F\setminus E(G^{\mathcal{M}}[A^{\sim}])$ is a minimal set of pairs such that every \Pvalid{} matching in~$\mathcal{M}$ contains an edge from the set.
	
	As $e \in M'$, \Cref{lem:A-sim} further implies
	that the endpoints of~$e$ are both contained in~$A^\sim$.
	This implies that $ F\setminus E(G^{\mathcal{M}}[A^{\sim}]) \subsetneq F$.
	Let $M''\in \mathcal{M}$ be a \Pvalid{} matching containing pair $e'$ (which is guaranteed to exist by the minimality of $F$).
	If both endpoints of~$e'$ would be contained in~$A^\sim$, then as $G^{\mathcal{M}}[A^{\sim}] = G^{\mathcal{M}'} [A^{\sim}]$,
	matching~$(M''\cap  E(G^{\mathcal{M}}[A^{\sim}]))\cup (M' \cap E(G^{\mathcal{M'}}[A\setminus A^{\sim}]))$ is a complete matching in~$G^{\mathcal{M'}}$  (and thus part of $\mathcal{M}'$), which contains $e'$ and all forced pairs (as $M''$ and $M'$ contain all forced pairs), a contradiction to our initial assumption. 
	Thus, we have~$ F \cap E(G^{\mathcal{M}}[A^{\sim}])  \subsetneq F$.
	We reached a contradiction to the minimality of~$F$, proving the claim.
\end{claimproof}
By \Cref{lem:F}, 
for each closed and complete subset of rotations~$Z$, 
we have one of the following two possibilities 
\begin{enumerate*}[label=(\roman*)]
  \item for every pair~$e\in F$, there is some \Pvalid{} matching corresponding to $Z$ containing~$e$, or 
  \item for every pair $e\in F$, there is no \Pvalid{} matching corresponding to $Z$ containing~$e$.
\end{enumerate*}
 This implies that the necessary and \forrot\ rotation are the same for each pair of~$F$.
 \end{proof}

 Next, we show that for each such set~$F$ that any stable matching matches either all agents from~$F$ as good as in~$\mathcal{M}$ or different than in $\mathcal{M}$:
 \begin{lemma}
 \label{lem:strictly-other}
  Let~$F$ be a minimal set of pairs such that every \Pvalid{} matching of some equivalence class~$\mathcal{M}$ of stable matchings contains at least one pair of~$F$.
  Then for every equivalence class~$\mathcal{M}'$ of stable matchings which contains at least \Pvalid{} matching which is disjoint from~$F$, in all matchings from $\mathcal{M}'$ all agents incident to a pair from~$F$ are matched with a different rank than in~$\mathcal{M}$.
 \end{lemma}

 \begin{proof}
 Note that in case $\mathcal{M}$ does not contain a \Pvalid{} matching, $F$ is the empty set and thus the lemma statement is trivially fulfilled.
Thus, assume in the following that $\mathcal{M}$ contains at least one \Pvalid{} matching.

  Let~$M \in \mathcal{M}$ and $M' \in \mathcal{M}'$.
  Assume towards a contradiction that there is an agent~$a$ incident to a pair~$e \in F$ with~$M(a) \sim_a M' (a)$.
  Let $A^\sim  := \{a \in A \mid M(a) \sim_{a} M'(a)\}$.
  By \Cref{lem:A-sim},  each stable matching from~$\mathcal{M} \cup \mathcal{M}'$ contains a complete matching from~$A^\sim$ and one on the remaining agents.
  Since $\mathcal{M}'$ contains a stable matching disjoint from~$F$ containing all forced pairs but $\mathcal{M}$ does not contain a stable matching disjoint from $F$ containing all forced pairs, it follows that $F \subsetneq E(G[A^\sim])$.
  \Cref{lem:A-sim} implies that $ F \cap E(G^{\mathcal{M}}[A^{\sim}])$ or $ F\setminus E(G^{\mathcal{M}}[A^{\sim}])$ is a minimal set of pairs such that every stable matching containing all forced pairs~$P$ in~$\mathcal{M}$ contains a pair from the set.
  Since~$e \in E(G[A^\sim])$, it follows that $F\cap E(G[A^\sim])$ or $F\setminus E(G[A^\sim])$ is a smaller minimal set of pairs for~$\mathcal{M}$, a contradiction to the minimality of~$F$.
 \end{proof}

\begin{lemma}\label{lem:forced-forbidden}
 Let $\mathcal{M}$ be an equivalence class containing a \Pvalid{} matching.
 Let $F$ be a minimal set of pairs such that every \Pvalid{} matching from $\mathcal{M}$ contains at least one pair of~$F$. 
 Assume that $\mathcal{M}$ also contains a stable matching~$M^d$ disjoint from~$F$.
 Then no \Pvalid{} matching is disjoint from~$F$.
\end{lemma}

\begin{proof}
 Assume towards a contradiction that there exists some \Pvalid{} matching~$M^*$ which is disjoint from~$F$.
 By the definition of $F$, matching~$M^*$ is not contained in~$\mathcal{M}$.
 Let $M$ be a \Pvalid{} matching from $\mathcal{M}$.
 By \Cref{lem:strictly-other}, no agent incident to an endpoint of a pair from~$F$ is indifferent between~$M^*$ and~$M$.
 Thus, for each pair $e$ from $F$, the rotation poset corresponding to $M^*$ either needs to contain the prohibited rotation corresponding to $e$ or does not include the necessary one.
 Since $M^*$ is \Pvalid, it follows that no forced pair has the same necessary and \forrot\ rotation as any pair from~$F$.
 Let~$X$ be the set of endpoints of pairs having the same necessary rotation and the same prohibited rotation as a pair from~$F$ (recall that the necessary and prohibited rotation are the same for all pairs from $F$ by \Cref{cor:nec-prh-rot-F}).
 Then any stable matching corresponding to~$\mathcal{M}$ contains a perfect matching on~$X$ by \Cref{lem:same-nec-prh}.
 
 Because matchings $M$ and $M^d$ corresponding to $\mathcal{M}$ contain a perfect matching on~$X$, we have that $M' := (M\setminus E(G[X])) \cup (E(G[X]) \cap M^d)$ is a matching.
 Clearly, $M'$ is contained in~$\mathcal{M}$.
 Further, $M'$ is disjoint from~$F$ as $M^d$ is disjoint from $F$ and each edge from~$F$ has one endpoint in~$X$.
 Observe that matchings~$M$ and $M'$ differ only in edges for which both endpoints are in~$X$, and all vertices of~$X$ are indifferent between~$M$ and $M'$.
 By \Cref{lem:same-nec-prh}, every edge from~$M \oplus M'$ has the same necessary and prohibited rotation as edges from $F$. 
 As we have observed that no forced pair has the same necessary and prohibited rotation as a pair from $F$, this implies that $M \oplus M'$ contains no forced pair.
 Using this and as $M$ contained all forced edges, it follows that $M'$ contains every forced pair.
 As we have already observed that $M'$ is disjoint from $F$, from this we reach a contradiction to the definition of $F$.
\end{proof}

 We have one last lemma before describing the algorithm:

\begin{lemma}
\label{lem:Z'irrelevant}
 Let $F$ be a set of edges with the same necessary~$\varphi_{\nec}$ and prohibited rotation~$\varphi_{\prh}$.
 Then either  \begin{enumerate*}[label=(\roman*)]
\item each \Pvalid{} matching which corresponds to some closed and complete subset~$Z$ of rotations containing~$\varphi_{\nec}$ and $\bar \varphi_\prh$ contains an edge from~$F$, or 
\item for each closed and complete subset~$Z$ of rotations containing $\varphi_{\nec}$ and $\bar \varphi_\prh$ for which there is some \Pvalid{} matching corresponding to~$Z$  there is also a \Pvalid{} matching corresponding to~$Z$ which is disjoint from~$F$.
\end{enumerate*}
\end{lemma}

\begin{proof}
We show that in case statement (i) is not true, statement (ii) is guaranteed to hold.
For this, assume that there is a closed and complete subset~$Z^*$ of rotations containing~$\varphi_{\nec}$ and $\bar \varphi_\prh$ such that there is a \Pvalid{} matching~$M^*$ corresponding to~$Z^*$ which is disjoint from~$F$.
 Let~$Z$ be a (different) closed and complete subset of rotations containing~$\varphi_{\nec}$ and $\bar \varphi_\prh$ and let~$M$ be a \Pvalid{} matching corresponding to~$Z$.
 We must show that there is a \Pvalid{} matching corresponding to~$Z$ which is disjoint from~$F$.

 Let~$A^{\sim}$ be the set of agents which are indifferent between~$M$ and~$M^*$.
 By \Cref{lem:A-sim}, $M' := \bigl(M \setminus E(G[A^{\sim}])\bigr) \cup \bigl(M^* \cap E(G[A^{\sim}])\bigr)$ is also a \Pvalid{} matching.
 Because each edge from~$F$ is contained in~$E(G[A^{\sim}])$ (as both $Z$ and $Z^*$ contain rotations $\varphi_{\nec}$ and $\bar \varphi_\prh$ ) and $M^*$ contains no pair from~$F$, it follows that $M'$ is a matching corresponding to $Z$ which does not contain a pair from~$F$.
\end{proof}

\subsection{The Algorithm}\label{strong:alg}

Using the machinery developed in the previous two subsections, we now present our algorithm to solve \ISRStrongtiesForcedForbidden. 
The general approach of the algorithm is similar as for our algorithm for the variant without ties (we tried to make both the description of the algorithm and its proof of correctness parallel to the case of strict preferences from \Cref{thm:fpt-new}). 
However, tackling the general variant comes at the cost of a slower running time. 
\SMstrong*
\begin{proof}
Recall that we can compute matchings corresponding to a closed and complete subset $Z$ of rotations by computing complete matchings in $G^Z$.
	We apply a similar algorithm as for \Cref{thm:fpt-new}:
	\begin{enumerate}
		\item \label{item:poset-strong}
		Compute the rotation poset of $\mathcal{I}$.
		Let $Z_1$ be the closed complete subset of rotations corresponding to the equivalence class containing~$M_1$ (\Cref{poset:equiv}).
		Set $Z:= Z_1$.
		
		\item \label{item:forced-new}
		For each forced pair $e = \{a,b\}\in Q$ that is not a fixed pair, we
		integrate~$\varphi^{e, a}_{\nec}$ and $\bar \varphi^{e, a}_{\prh}$ to~$Z$ (if existent).
		If no \Pvalid{} matching corresponding to~$Z$ exists, then return that there is no \Pvalid{} matching.
		
		\item[(2.5)] 
		Let~$Z_1^*$ be the current set of rotations and $M_1^*$ a \Pvalid{} matching corresponding to~$Z_1^*$ and containing as many pairs from~$M_1$ as possible.
		Let $P^{\reduced} := \emptyset$.
		For each set~$F$ of forbidden pairs with the same necessary rotation~$\rho$ and the same \forrot\ rotation~$\phi$, we pick an arbitrary closed and complete subset~$Z'$ of rotations containing~$\rho$ and~$\bar \phi$ as well as the necessary but not the \forrot\ rotation for each forced pair and check whether there is a stable matching corresponding to~$Z'$ containing no pair of~$F$.
		If this is the case, then do nothing.
		Otherwise, we add either a pair from~$F \cap M_1$ (if $F\cap M_1 \neq \emptyset$) or an arbitrary pair from~$F$ (if $F\cap M_1 = \emptyset$) to~$P^{\reduced}$.
		
		\item \label{item:forbidden-new}
		For each pair~$e = \{a, b\} \in P^{\reduced} \cap M_1$, do the following:
		If $a$ and $b$ are indifferent between~$M_1$ and $M_1^*$, then guess whether~$a$ or~$b$ shall strictly prefer~$M_2$ to~$M^*_1$.
		Otherwise $a$ or $b$ prefers~$M_1^*$ to~$M_1$.
		Without loss of generality assume that we either guessed that~$a$ prefers~$M_2$ to~$M_1^*$ or that $a$ prefers~$M_1^*$ to~$M_1$.
		Integrate $\bar \varphi^{e, a}_\nec$ into~$Z$.
		(If $\varphi^{e, a}_\nec$ does not exist, then reject the guess.)
		
		\item \label{item:forbidden-new-notM1} 
        Pick a forbidden pair~$e=\{a, b\} \in P^{\reduced}$ such that $\varphi^{e, a}_\nec \in Z$ and $\varphi^{e, a}_\prh \notin Z$ (recall that this is equivalent to $\varphi^{e, b}_\nec \in Z$ and $\varphi^{e, b}_\prh \notin Z$).
		If there is no such forbidden pair, then go to \Cref{item:return}.
		If $a$ is indifferent between~$M_1$, $M_1^*$, and~$b$, then reject this guess.
		Without loss of generality assume that $a$ weakly prefers~$b$ to~$M_1^*$ and weakly prefers~$M^*_1$ to~$M_1$.
		Integrate $\bar \varphi^{e, a}_\nec$ into~$Z$.
		If $\varphi^{e, a}_\nec\in Z\setminus Z_1$ , then reject this guess.
		Repeat this step.
		
		\item \label{item:return}
		Return the \Pvalid{} matching corresponding to~$Z$ that contains as many pairs as possible from~$M_1$.
	\end{enumerate}

	\paragraph{Proof of Correctness.}
	
	We start by showing that all changes made to $Z$ over the course of the algorithm can indeed be done and are indeed necessary.
	\begin{claim}\label{claim:strong-FPT-algo}
		Let~$M^*$ be a \Pvalid{} matching containing no forbidden pairs which respects our guesses.
		Further, let $Z^*$ be the corresponding closed and complete subset of rotations.
		Then $Z^*$ contains all rotations added in \Cref{item:forced-new,item:forbidden-new-notM1,item:forbidden-new}.
		Moreover, in \Cref{item:forbidden-new}, pair~$\{a,b\}$ always contains an agent strictly preferring $M_2$ to $M_1^*$, and, in \Cref{item:forbidden-new-notM1}, pair~$\{a, b\}$ always contains an agent weakly preferring the other endpoint from the pair to~$M_1^*$ and weakly preferring~$M^*_1$ to~$M_1$.
		Further, from \Cref{item:forced-new} on, there is always a \Pvalid{} matching corresponding to $Z$.
	\end{claim}
	
	\begin{claimproof}
	We prove all parts of the claim simultaneously. For the first part, it suffices to prove the statement for all integrated rotations.
        
        \begin{enumerate}
         \item[(2)] 
		Each rotation integrated in \Cref{item:forced-new} is contained in $Z^*$ by \Cref{le:nec-phr}.
		If the algorithm does not return that there is no solution, then further clearly there is some \Pvalid{} matching corresponding to~$Z_1^*$.
		So assume towards a contradiction that the algorithm rejects in \Cref{item:forced-new} but there is some \Pvalid{} matching~$N$.
		This can happen for two reasons:
		First, the necessary rotation (or the dual of the forbidden rotation) of some forced pair has been eliminated while integrating some $\varphi^{e, a}_{\nec}$ or $\bar \varphi^{e, a}_{\prh}$ for some other forced pair, implying that no \Pvalid{} matching can exist.
		Second, assume that the first case never happened (but there is still no \Pvalid{} matching corresponding to~$Z^*$).
		Then each endpoint of~$Q$ is indifferent between any matching~$M^*$ corresponding to~$Z^*$ and $N$.
		Let~$A^{\sim}$ be the set of agents which are indifferent between~$M^*$ and $N$.
		By \Cref{lem:A-sim}, it follows that $\bigl(M^* \cap E(G[V\setminus A^{\sim}])\bigr) \cup \bigl(N\cap E(G[A^{\sim}])\bigr)$ is a \Pvalid{}matching corresponding to~$Z_1^*$, a contradiction.

		\item[(3)]
		Next, we consider the rotations integrated in \Cref{item:forbidden-new}.
	First, we consider the case that $a$ and $b$ are indifferent between~$M_1$ and $M_1^*$ and prove that as a consequence of \Cref{cor:nec-prh-rot-F,lem:strictly-other}, for each pair~$e =\{a, b\} \in P^{\reduced} \cap M_1$, either~$a$ or $b$ must prefer~$M_2 $ to~$M_1^*$, justifying our guessing:
	Let~$F^e$ be the set of forbidden pairs with the same necessary and \forrot\ rotation as~$e$.
	Then by the definition of~$P^{\reduced}$, there is a minimal subset~$F$ of~$F^e$ such that each \Pvalid{} matching corresponding to~$Z'$ (where $Z'$ is the set from Step~(2.5) constructed for $F^e$) contains a pair from~$F$.
	Consequently, by \Cref{lem:strictly-other}, every \Pvalid{} matching disjoint from~$F$ matches each endpoint of a pair from~$F$ to a different rank than it is matched in~$M_1$ and $M^*_1$ (note that by our initial assumption all agents from $F^e$ are indifferent between $M_1$ and $M^*_1$).
	By \Cref{le:nec-phr} and as $M_2$ cannot contain pairs from $F$, it follows that we need to integrate the \forrot\ rotation or the dual of the necessary rotation of pairs from $F^e$ to~$Z$:
	By the definition of~$F^e$, this changes the rank that agents~$a$ and~$b$ have for their partner in the matching (implying that their rank is different in $M_1^*$ and $M_2$).
	By \Cref{obs:indifference}, one of~$a$ and $b$ prefers~$M_1^*$ to~$M_2$ while the other prefers~$M_2$ to~$M_1^*$.
	
	We now continue by arguing that all changes made to $Z$ are indeed necessary. 
	If $a$ and $b$ are indifferent between~$M_1 $ and $M_1^*$, then we assume without loss of generality that we guess that $a$ prefers $M_2(a)$ to $M_1^*(a)=b$.
	If $a$ and $b$ are not indifferent between~$M_1$ and $M_1^*$ (without loss of generality assume that $a$ prefers~$M_1^*$ to~$M_1$), then $a$ prefers~$M_2$ to~$M_1^*$ by \Cref{lem:strictly-other} and as all changes so far were necessary.
	In both cases, 
		$a$ must be matched strictly better than~$b$ in the desired matching and we can proceed as follows.
		If $a$ has no stable partner it prefers to~$b$, then $\varphi_\nec^{e, a}$ does not exist and the algorithm correctly rejects.
		So assume that $\varphi_{\nec}^{e, a}$ exists.
		As $a$ must be matchted strictly better than~$b$, set~$Z^*$ cannot contain a rotation $(a^m,x,y)$ with $y\geq \rk_{a^m}(b^w)$ and thus in particular not $\varphi_{\nec}^{e, a}$ (see \Cref{poset:equiv}).
		Thus $\bar \varphi_{\nec}^{e, a} \in Z^*$. 
        
      \item[(4)]
        Let $M$ be any strongly stable matching corresponding to~$Z$ which contains~$e$ (such a matching has to exist as $\varphi_\nec^{e,a} \in Z$ and $\varphi_\prh^{e,a} \notin Z$).
        We first show that if $a$ and $b$ are indifferent between~$M$, $M_1^*$, and $M_1$, then the algorithm correctly rejects the guess.
        Note that as both $a$ and $b$ are indifferent between $M$, $M_1^*$, and $M_1$ and as $M$ contains $e$ (implying that $Z$ contains $\varphi_\nec^{e,a} \in Z$ and $\varphi_\prh^{e,a} \notin Z$), by \Cref{le:nec-phr} also $Z_1$ and $Z^*_1$  contain the necessary but not the \forrot\ rotation for~$e$.
        Again, let~$F^e$ be the set of forbidden pairs with the same necessary and \forrot\ rotation as~$e$. 
        By the definition of~$P^{\reduced}$, for the closed and complete subset $Z'$ from Step~(2.5) constructed for $F^e$, every \Pvalid{} matching corresponding to~$Z'$ contains at least one pair from~$F^e$.
        Since $Z_1^*$ and $Z'$ contain the necessary but not the \forrot\ rotation for~$e$, it follows that each stable matching corresponding to~$Z_1^*$, and in particular~$M_1^*$, contains at least one pair~$e'$ from~$F^e$ (if $M_1^*$ is disjoint from $F^e$, then from \Cref{lem:strictly-other} we get that both endpoints of one of the pairs from $F^e$ have a different rank for their partner in $M_1^*$ and matchings corresponding to $Z'$, implying that $Z_1^*$ does not contain the necessary or contains the prohibited rotation of pairs from~$F^e$, a contradiction).
        Let $A^{\sim} := \{ a \in A \mid M_1 (a ) \sim_a M_1^*( a)\}$ and 
        $N_1^* := \bigl(M_1^* \cap E(G[V\setminus A^{\sim}])\bigr) \cup \bigl(M_1\cap E(G[A^{\sim}])\bigr)$.
        By \Cref{lem:A-sim},~$N_1^*$ is a matching. 
        It follows that~$N_1^*$ is a stable matching corresponding to~$Z_1^*$ and containing all pairs from~$M_1 $ whose endpoints are indifferent between~$M_1^*$ and $M_1$.
        If there exists some~$ e'\in N_1^* \cap F^e \subseteq M_1 \cap F^e$, then by \Cref{lem:Z'irrelevant} (and as from $e'\in P^{\reduced}$ it follows that there is a closed and complete subset of rotations containing $\varphi_\nec^{e,a}$ and $\bar \varphi_\prh^{e,a}$ for which no \Pvalid{} matching is disjoint from $F^e$ and thus statement (ii) from \Cref{lem:Z'irrelevant} cannot be fulfilled) there is also some edge~$e'' \in M_1 \cap F^e$ and we would have added $e''$ to~$P^{\reduced}$ in Step~(2.5) and handled~$e''$ in \Cref{item:forbidden-new}.
        Thus, we may assume that $N_1^* \cap F^e = \emptyset$.
        It follows that $N_1^*$ does not contain all forced pairs (otherwise we would not have added~$e $ to~$P^\reduced$ in Step~(2.5)).
        By \Cref{lem:forced-forbidden}, it follows that there is no \Pvalid{} matching containing no forbidden pairs.
        Consequently, our algorithm correctly rejects the guess.

        Now consider the case that one of the endpoints, say~$a$, prefers~$M$ to~$M_1^*$ or $M_1^*$ to~$M_1$.
        As the changes done so far are necessary, it follows that $M (a) \succsim M_1^* (a) \succsim M_1 (a)$.
        Since the changes done so far are necessary, it follows by \Cref{thm:strong-rotation-equivalence} that $a$ must be matched at least as good as~$b$ in a stable matching corresponding to~$Z^*$. 
        \Cref{lem:strictly-other} implies that $a$ must be matched better than $b$ in a stable matching corresponding to~$Z^*$.
        From here on, the proof is analogous to \Cref{item:forbidden-new}.
      \end{enumerate}
      This finishes the proof of the claim.
	\end{claimproof}
	
	Using \Cref{claim:strong-FPT-algo}, we finish the proof of correctness as follows:
    First observe that if the algorithm rejects a guess in \Cref{item:forbidden-new-notM1}, then there is no \Pvalid{} matching containing no forbidden pair and this guess due to \Cref{claim:strong-FPT-algo} and \Cref{thm:strong-rotation-equivalence}.
    Therefore, assume from now on that the current guess was not rejected.
    Thus, there is no rotation~$\rho$ such that $\rho$ as well as $\bar \rho$ get added to~$Z$ during the algorithm.
	\Cref{item:forced-new} now ensures by \Cref{le:nec-phr} that $M$ contains all forced pairs.
	\Cref{item:forbidden-new-notM1} ensures that $M$ contains no forbidden pair.
	
	Next, we show the optimality of the returned matching~$M$.
	Let~$Z^*$ be the subset of the rotation poset corresponding to an optimal stable matching~$M^*$ (if no stable matching obeying our guesses exists, then we reject all guesses).
	By \Cref{claim:strong-FPT-algo}, we get 
	$Z_1 \triangle Z \subseteq Z_1 \triangle Z^*$ (\Cref{claim:strong-FPT-algo} directly implies that $Z\setminus Z_1\subseteq Z^*\setminus Z_1$ but also gives us $Z_1\setminus Z\subseteq Z_1\setminus Z^*$ as deleting a rotation corresponds to adding its dual).
	We now show that we can conclude from this that $|M\cap M_1| \ge |M^* \cap M_1|$:
	Let $A^\sim$ be the set of agents indifferent between $M$ and $M^*$.
	We claim that any pair from~$M^*\cap M_1$ contains only agents from~$A^\sim$.
	Assume towards a contradiction that 
	there is some pair~$e = \{a, b\} \in M_1 \cap M^*$ with~$a \notin A^\sim$.
	Without loss of generality assume that $a$ prefers~$M^* (a) = M_1 (a)$ to~$M (a)$.
	Then the dual of the necessary or the \forrot\ rotation of~$e$ has to be contained in~$Z$.
	Since $Z_1 \triangle Z \subseteq Z_1 \triangle Z^*$, it follows that the dual of the necessary rotation or the \forrot\ of~$e$ is contained in~$Z^*$, contradicting~$e \in M^*$.
	As every pair from~$M^* \cap M_1$ has both endpoints in $A^\sim$, it follows that $M':= (M \setminus G[A^\sim]) \cup (M^* \cap G[A^\sim])$ is a strongly stable matching corresponding to~$Z$ with $|M' \cap M_1| > |M \cap M_1|$, a contradiction to the definition of~$M$.

\paragraph{Running Time.}
  \Cref{item:poset-strong} can be done in $\mathcal{O} (n m)$ time~\cite{DBLP:conf/soda/KunyszPG16}.
  The necessary and \forrot\ of a pair~$e$ can be computed in constant time.
  For each set~$F$, we can construct~$Z^*$ in $\mathcal{O}(m)$ time.
  Testing whether there is a stable matching corresponding to~$Z^*$ and disjoint from~$F$ can be done in $\mathcal{O} (\sqrt{n}m \log n)$ time~\cite{DBLP:journals/talg/DuanPS18} by computing a maximum-weight complete matching.
  Thus, Step~(2.5) can be done in $\mathcal{O}(\sqrt{n} m^2 \log n)$ time.
  Since any rotation gets at most once integrated, it follows that all integrations of rotations together take $\mathcal{O}(m)$ time.
  \Cref{item:forced-new,item:forbidden-new-notM1} only integrate rotations and thus runs in $\mathcal{O}( m)$ time.
  \Cref{item:forbidden-new} makes up to~$|P \cap M_1| \le |P|$ many guesses.
  Apart from this, only rotations get integrated.
  Thus, its running time is $\mathcal{O}(2^{|P \cap M_1|}  \cdot m)$.
  Thus, \Cref{item:forbidden-new,item:forbidden-new-notM1,item:forced-new} can be done in $\mathcal{O}( m )$ time per guess.
  \Cref{item:return} computes one maximum-weight matching and thus runs in $\mathcal{O}(\sqrt{n} m \log n)$ time.
  Consequently, the total running time is $\mathcal{O} \bigl( (2^{|P\cap M_1|} + m) \cdot \sqrt{n} m \log n\bigr)$.
\end{proof}

\end{document}